\documentclass[11pt]{article}

\usepackage{amsmath,amsfonts,amssymb,amsthm,braket}
\usepackage[margin=0.63in]{geometry}
\usepackage[colorlinks,
            linkcolor = blue,
            urlcolor  = blue,
            citecolor = blue,
            anchorcolor = blue]{hyperref}
\usepackage{mathrsfs}
\usepackage{mathtools}
\usepackage{thmtools}
\usepackage{thm-restate}
\mathtoolsset{centercolon}
\usepackage[numbers,comma,sort&compress]{natbib}
\usepackage{bm}
\usepackage{bookmark}
\usepackage{svg}
\usepackage[font=small]{caption}
\usepackage{titlesec}
\usepackage{enumitem}

\titlespacing*{\paragraph}{0pt}{4pt plus 1pt minus 1pt}{.75em}

\let\originalleft\left
\let\originalright\right
\renewcommand{\left}{\mathopen{}\mathclose\bgroup\originalleft}
\renewcommand{\right}{\aftergroup\egroup\originalright}

\newcommand{\BQP}{{\mathsf{BQP}}}
\newcommand{\BPP}{{\mathsf{BPP}}}
\newcommand{\C}{{\mathbb{C}}}
\newcommand{\R}{{\mathbb{R}}}

\newcommand{\I}{\mathscr{I}}

\DeclareMathOperator{\poly}{poly}
\DeclareMathOperator{\polylog}{polylog}

\newtheoremstyle{tight}
  {1.5pt} 
  {1.5pt} 
  {\itshape} 
  {} 
  {\bfseries} 
  {.} 
  {.5em} 
  {} 
\theoremstyle{tight}
\newtheorem{theorem}{Theorem}
\newtheorem{lemma}{Lemma}
\newtheorem{corollary}{Corollary}
\newtheorem{proposition}{Proposition}
\newtheorem{problem}{Problem}
\newtheorem{remark}{Remark}
\newtheorem{definition}{Definition}

\newcommand{\eq}[1]{(\ref{eq:#1})}
\newcommand{\alg}[1]{\hyperref[alg:#1]{Algorithm~\ref*{alg:#1}}}
\newcommand{\defn}[1]{\hyperref[defn:#1]{Definition~\ref*{defn:#1}}}
\renewcommand{\sec}[1]{\hyperref[sec:#1]{Section~\ref*{sec:#1}}}
\newcommand{\thm}[1]{\hyperref[thm:#1]{Theorem~\ref*{thm:#1}}}
\newcommand{\lem}[1]{\hyperref[lem:#1]{Lemma~\ref*{lem:#1}}}
\newcommand{\cor}[1]{\hyperref[cor:#1]{Corollary~\ref*{cor:#1}}}
\newcommand{\prb}[1]{\hyperref[prb:#1]{Problem~\ref*{prb:#1}}}
\newcommand{\prop}[1]{\hyperref[prop:#1]{Proposition~\ref*{prop:#1}}}
\newcommand{\fig}[1]{\hyperref[fig:#1]{Figure~\ref*{fig:#1}}}
\newcommand{\appx}[1]{\hyperref[appx:#1]{Appendix~\ref*{appx:#1}}}

\makeatletter
\let\@@magyar@captionfix\relax
\makeatother

\usepackage[utf8]{inputenc}
\usepackage{color,soul}
\usepackage{graphicx, fullpage, dsfont}
\usepackage{authblk}

\usepackage{xcolor}
\usepackage{pgfplots,caption}
\pgfplotsset{compat=1.17}
\usepackage{fullpage}
\usepackage{doi}

\numberwithin{equation}{section}

\newcommand{\be}{\begin{equation}}
\newcommand{\ee}{\end{equation}}

\begin{document}

\title{Quantum algorithms to simulate quadratic classical Hamiltonians and optimal control}
\author{Hari Krovi }
\date{\today}
\affil[]{ Riverlane Research, Cambridge, MA}

\maketitle

\begin{abstract}
Simulation of realistic classical mechanical systems is of great importance to many areas of engineering such as robotics, dynamics of rotating machinery and control theory. In this work, we develop quantum algorithms to estimate quantities of interest such as the kinetic energy in a given classical mechanical system in the presence of friction or damping as well as forcing or source terms, which makes the algorithm of practical interest. We show that for such systems, the quantum algorithm scales polynomially with the logarithm of the dimension of the system. We cast this problem in terms of Hamilton's equations of motion (equivalent to the first variation of the Lagrangian) and solve them using quantum algorithms for differential equations. We then consider the hardness of estimating the kinetic energy of a damped coupled oscillator system. We show that estimating the kinetic energy at a given time of this system to within additive precision is $\BQP$ hard when the strength of the damping term is bounded by an inverse polynomial in the number of qubits. In other words, when combined with prior work on coupled oscillators, this shows that we can get a quantum advantage even in the presence of damping.

We then consider the problem of designing optimal control of classical systems, which can be cast as the second variation of the Lagrangian. In this direction, we first consider the Riccati equation, which is a nonlinear differential equation ubiquitous in control theory. We give an efficient quantum algorithm to solve the Riccati differential equation well into the nonlinear regime. To our knowledge, this is the first example of any nonlinear differential equation that can be solved when the strength of the nonlinearity is asymptotically greater than the amount of dissipation. We then show how to use this algorithm to solve the linear quadratic regulator problem, which is an example of the Hamilton-Jacobi-Bellman equation. Our results show that quantum algorithms for differential equations, especially for ordinary linear differential equations, can go quite far in solving problems of practical relevance.
\end{abstract}

\section{Introduction}\label{sec:introduction}
Quantum algorithms have been designed for several important problems such as cryptography, quantum Hamiltonian simulation, condensed matter physics, semi-definite programming and optimization. Important primitives have been developed such as quantum phase estimation, quantum walks, Fourier transforms, Hamiltonian simulation and quantum linear systems algorithms (QLSA). The latter two have a close relationship and common techniques can be used in implementing them. In fact, it was shown in \cite{10.1145/3313276.3316366, PRXQuantum.2.040203} that one can include quantum walks, phase estimation and Hamiltonian simulation to show that they can all be cast as the problem of implementing functions of Hamiltonians. This has recently been generalized to a large set of non-Hermitian matrices in \cite{low2024quantum}. Hamiltonian simulation essentially solves the Schr\"{o}dinger equation, which is a linear ordinary differential equation (ODE). Efficient techniques for this have been developed in a series of papers \cite{Lloyd_Ham_sim, BAC07,Ham_sim_opt, Low2019hamiltonian, chakraborty_et_al:LIPIcs:2019:10609} with current techniques having an optimal dependence on the parameters. Hamiltonian simulation algorithms for time-dependent Hamiltonians have been developed in \cite{PhysRevA.99.042314,Berry2020timedependent, Kalev2021quantumalgorithm}. 

The HHL algorithm to solve linear systems of equations was described in the seminal paper \cite{HHL09}. This algorithm has been generalized in subsequent works (see \cite{dervovic2018quantum} for a review), where the matrix representing the linear system can be non-Hermitian, rectangular and even singular. The dependence on various parameters such as the condition number and error has been improved in \cite{ambainis:LIPIcs:2012:3426,CKS15,QLSA_linear_kappa}. This algorithm has been applied to problems in topological data analysis \cite{LGZ,TDA}, machine learning \cite{PCA} and more recently in solving differential equations (described in detail below). An important subroutine developed in the process is a protocol to load classical data called block-encoding. This method takes classical data such as matrix entries of sparse and structured matrices and creates a unitary that represents the matrix in the quantum algorithm. Block-encoding has been used in earlier quantum algorithms such as \cite{CKS15} and recently many explicit methods have been developed \cite{10.1145/3313276.3316366,chakraborty_et_al:LIPIcs:2019:10609, vanApeldoorn2020quantumsdpsolvers, Sunderhauf2024blockencoding, 10012045}. We use these techniques in this paper as well.

An approach to solve differential equations using the linear systems algorithms was given in \cite{Ber14}. This was made into a high-precision algorithm in \cite{BCOW17} using truncated Taylor series to approximate $\exp(At)$. These algorithms have been generalized to allow for arbitrary time-independent matrices in \cite{krovi2022improved}. An important quantity used in this work is the norm of the exponential of the matrix i.e., $\|\exp(At)\|$. It was shown that as long as this norm of the exponential of the matrix $A$ over the time domain of interest (denoted $C(A)$ in \cite{krovi2022improved}) is polynomially bounded, one can bound the condition number of the resulting linear system. Using a stronger condition that the log-norm of $A$ is negative, this has been extended to time-dependent matrices in \cite{time_marching, berry2022quantum}. More recently, an approach to reduce time-dependent differential equations to Hamiltonian simulation is given in \cite{an2023quantum, LCHS} where optimal dependence on parameters is achieved assuming $A$ has a negative log-norm. It is interesting to note that for one of the applications considered here (of damped coupled oscillators), the matrix involved in the differential equation does not have a negative log-norm and needs the weaker condition of bounded $C(A)$. A general theory of quantum algorithms for linear ODEs is given in \cite{ode_theory} where a fast-forwarding theorem similar to the one in Hamiltonian simulation is proved. Recently, there has been work on using Hamiltonian simulation algorithms to solve general ODEs. In \cite{jin2022quantum}, the authors reduce a general class of partial differential equations to Hamiltonian simulation.

Many of the above algorithms essentially solve a \emph{quantum} problem i.e., given an initial state as a quantum state, and sparse access oracles, the goal is to create a quantum state that encodes the normalized solution state. However, in practical applications, one usually has access to classical data that must be encoded into the quantum state and at the end of the simulation we would also need to extract classical data out of the quantum state. Many quantum algorithms are designed within a specific mathematical framework and applying them to problems of practical interest is not always straightforward. The comprehensive review \cite{dalzell2023quantum} covers the above mentioned algorithms and many more from the point of view of end-to-end applications. In considering such end-to-end applications, one needs proof or strong evidence of three different aspects namely, quantum efficiency, classical hardness and practical utility.

The recent work on estimating the kinetic energy of coupled oscillators \cite{babbush2023exponential} is an example of such an end-to-end application. In that paper, the authors consider a problem of practical interest and show that quantum algorithms for Hamiltonian simulation can be used to provide evidence of a quantum advantage. The specific problem considered is that of estimating kinetic energy of a network of coupled oscillators. The problem of estimating the kinetic energy (or more precisely, the decision version where one has to decide whether the kinetic energy is above some value or below another with the promise that the gap is a constant independent of the number of qubits) has been shown to be $\BQP$ complete and therefore, one does not expect an efficient classical algorithm unless $\BQP=\BPP$. While this problem has a quantum advantage, an important caveat is that the oscillators are assumed to be ideal and there is no damping in the system. Going beyond this idealized setting i.e., simulating oscillators in the presence of damping would increase the practical utility of the quantum algorithm. In this work, we show that even in the presence of damping in coupled oscillators (or dissipation in other systems that can be modeled by the similar equations), there exist efficient quantum algorithms to approximate the kinetic energy of the system.

One might wonder if the presence of damping makes the problem of approximating the kinetic energy of an oscillator network amenable to classical algorithms i.e., it may be that the problem of kinetic energy estimation of damped coupled oscillators is in $\BPP$. We show that this is unlikely to be the case. Specifically, we show that even in the presence of damping, approximating the kinetic energy of an arbitrary sparse oscillator network is $\BQP$ hard (provided that the strength of damping is inverse polynomial in the number of qubits). This means that even under the realistic conditions, this problem has a quantum advantage. In order to construct a quantum algorithm for this problem, we cast this problem as ordinary differential equation (ODE) and use the recently developed quantum algorithm for ODEs \cite{krovi2022improved}.

In order to apply quantum algorithms for linear ODEs of the form $\dot{x}=Ax+b$, one needs to overcome certain obstacles. As pointed out in \cite{babbush2023exponential}, if the quantum algorithm for a linear ODE requires that $A$ has a negative log-norm, then one cannot get an efficient algorithm. The algorithm in \cite{krovi2022improved} requires that the norm of the exponential of $A$ is bounded and works even when the log-norm of $A$ is positive. However, this encoding means that the norm of the exponential of $A$ is proportional to the condition number of the matrix $V$, which encodes the spring constants. The algorithm of \cite{babbush2023exponential} did not depend on the condition number of $V$, but rather only the norm of $V$. This is the price one has to pay to be able to simulate damping (which is not amenable to Hamiltonian simulation). Fortunately, the work of \cite{babbush2023exponential} also showed the the hardest cases of the coupled oscillator problem (i.e., the instances which are $\BQP$ hard) have a constant condition number for $V$. This also carries over for the case with damping. We also show that in a different basis, we can design a quantum algorithm to estimate the kinetic energy with a slightly worse dependence on error, but no dependence on the condition number of $\sqrt{V}$. Therefore, depending on the condition number of $\sqrt{V}$, one can choose one of these bases. These algorithms are applicable to systems other than just damped forced coupled oscillators. For the specific case of damped and forced coupled oscillators, we show that one can remove the dependence on the condition number of $V$ and the worse dependence on error by using the specific structure of the problem (by leveraging the results for block-encoding in \cite{babbush2023exponential}).

Quantum algorithms for nonlinear differential equations have been constructed in \cite{Liue2026805118, krovi2022improved,Xue_2021,LPG20} and quantum algorithms for partial differential equations (PDEs) in \cite{CLO20, jin2022b, heat_eq}. For nonlinear ODEs, quantum algorithms such as \cite{Liue2026805118, krovi2022improved} can be applied to very general types of nonlinear ODEs (with a polynomial nonlinearity). However, these algorithms can only handle a small amount of nonlinearity i.e., they require $R<1$, where $R$ is (roughly) the ratio of the strength of the nonlinearity to the strength of dissipation. In \cite{Liue2026805118}, it was also shown that no quantum algorithm can efficiently solve nonlinear ODEs for $R>\sqrt{2}$. This was improved to $R\geq 1$ in \cite{lewis2023limitations} matching the algorithm. While this is true for arbitrary nonlinear differential equations, one might hope that for specific nonlinear ODEs, higher nonlinearity may be efficiently solvable using quantum algorithms. However, there are currently no quantum algorithms that can solve nonlinear ODEs for $R>1$. In this paper, we remedy this situation and show that a specific nonlinear differential equation called the Riccati equation used in optimal control can be solved for $R=\polylog(N)$, where $N$ is the dimension of the matrices involved in the equation. The solution to the Riccati equation, which is ubiquitous in classical optimal control, has practical applications such as the stability of magnetohydrodynamics \cite{Glasser}.

The solution to the Riccati equation can also be used to solve another problem in optimal control known as the linear quadratic regulator problem (LQR). This problem is about finding a controller that can regulate a process such as in an airplane by minimizing a cost function. It turns out that this can be cast as a Riccati equation by viewing it as a special case of the Hamilton-Jacobi-Bellman (HJB) equation. The HJB equation is a non-linear partial differential equation that generalizes the Hamilton-Jacobi equation from the theory of calculus of variations. It gives the necessary and sufficient conditions for optimality of a control problem. In its generality, this problem can be highly nonlinear. However, in the case of the LQR problem, it becomes tractable using quantum algorithms.

The differential equations used in classical mechanical systems described above and those in optimal control are linked via a theory known as Calculus of Variations \cite{gelfand2000calculus, liberzon2012calculus} as mentioned above. This theory has two equivalent approaches - a Lagrangian formalism and a Hamiltonian formalism. We use the Hamiltonian formulation of this theory since it gives us linear ODEs. We can then apply quantum algorithms for linear ODEs to solve the above problems of simulating classical mechanical systems under realistic conditions and of solving the Riccati and Hamilton-Jacobi-Bellman equation (for the LQG). Interestingly, in the above applications, we need the linear ODE solvers that are based on linear systems rather then Hamiltonian simulation. This is because we need to extract classical information from the unnormalized final state (which can be done via the so-called history state that keeps track of the solution at various times). Recent approaches via Hamitonian simulation directly produce the normalized final state and it would be interesting to see if they can be modified for these applications.

The rest of the paper is organized as follows. In \sec{prelims}, we give a basic background in calculus of variations, the Hamiltonian formalism, optimal control, Riccati equation and the HJB equation. Then in \sec{comp_prelims}, we present some background and definitions needed on oracle access and block encoding as well as results on quantum linear systems algorithms that we will need later on. In \sec{class_ham_sim}, we give our quantum algorithm to estimate the kinetic energy of damped coupled oscillators and other realistic mechanical systems. We also give a proof of the $\BQP$ hardness of additive approximation of the kinetic energy. Then in \sec{optimal_control}, we give quantum algorithms to solve the vector and matrix versions of the Riccati equation and its application to the LQR problem. Finally, in \sec{conclusions}, we present our conclusions and open questions.

\section{Preliminaries}\label{sec:prelims}

\subsection{Calculus of variations}\label{sec:Ham}
The starting point of the theory of calculus of variations is the Lagrangian which is a functional that depends on the position $q$, the velocity $\dot{q}$ and time $t$. The quantities $q$ and $\dot{q}$ may also depend on time. The goal is to find the minimum of the following integral of the Lagrangian.
\begin{equation}
    J(q)=\int_{t_i}^{t_f} L(q,\dot{q},t) dt\,,
\end{equation}
subject to some initial and final conditions on $q(t_i)$ and $q(t_f)$. In classical mechanics, this is called the \emph{principle of least action}. The choice of $q$ and $\dot{q}$ is not unique. Using canonical transformations, one can move to other coordinate systems such as polar coordinates. In this paper, we use the above ``Cartesian" coordinates so that the solution corresponds to physical positions and velocities. 

For the above minimization problem, one can derive a necessary condition similar to the condition that the vanishing at a point of the first derivative of a function gives its extremum. For the functional above, this is equivalent to the vanishing of the first variation, which leads to the famous Euler-Lagrange equations of motion. These equations can be written as
\begin{equation}
    \frac{\partial L}{\partial q_i} = \frac{d}{dt}\frac{\partial L}{\partial \dot{q}_i}\,,
\end{equation}
for all coordinate indices $i$. For example, when the mechanical system is a set of masses connected by springs, the Lagrangian describes a system of coupled oscillators. The equations of motion for this system were recently explored in \cite{babbush2023exponential}, where quantum algorithms for simulation and the complexity of estimating energies were given. The equations of motion in this case become
\begin{equation}
    M\ddot{q}(t) = Fq(t)\,,
\end{equation}
where $M$ is a diagonal matrix of masses and $F$ is a matrix that captures the force due to the spring constants.

There is a different set of equations, called Hamilton's canonical equations, that are first order unlike the Euler-Lagrange equations (which are second order equations). Hamilton's equations are equivalent to the Euler-Lagrange equations of motion and are derived from the classical Hamiltonian. To define the classical Hamiltonian, we first need to define the momentum $p$ which is conjugate to the position coordinate vector $q$ as follows.
\begin{equation}
    p_i=\frac{dL}{d\Dot{q}_i}\,,
\end{equation}
where $p_i$ are the components of $p$. Then one takes the Legendre transform to obtain the Hamiltonian as follows.
\begin{equation}
    \mathcal{H}(q,p,t)=p\dot{q} - L(q,\dot{q},t)\,,
\end{equation}
where $p\dot{q}$ is the dot product of $p$ and $\dot{q}$. The set of $2N$ coordinates $q_i$ and $p_i$ constitute the classical phase space of the mechanical system. For instance, the classical Hamiltonian of a system of coupled oscillators can be written as
\begin{equation}
    \mathcal{H} = \frac{1}{2}p^T M^{-1}p + q^T K q\,,
\end{equation}
where $K$ is the potential term of the spring constants. The equations of motion coming from the classical Hamiltonian, also known as Hamilton's canonical equations, can be written as
\begin{equation}\label{eq:Ham_eqs}
    \dot{q} = \nabla_p \mathcal{H}(q,p,t)\,,\hspace{0.1in} \dot{p} = -\nabla_q \mathcal{H}(q,p,t)\,,
\end{equation}
where
\begin{equation}
    \nabla_q = \begin{pmatrix}
        \partial_{q_1}\\ \vdots\\\partial_{q_N}
    \end{pmatrix}
    \,,
    \hspace{0.1in}
    \nabla_p = \begin{pmatrix}
        \partial_{p_1}\\ \vdots\\\partial_{p_N}
    \end{pmatrix}\,.
\end{equation}
This set of equations is equivalent to the Euler-Lagrange equations but are first order rather than second order equations. In the following, we present some known results \cite{de2006symplectic} on quadratic Hamiltonians that are relevant to this paper.

\subsection{Quadratic Hamiltonians}\label{sec:quad}
Suppose we have a quadratic Hamiltonian of the form
\begin{equation}
    \mathcal{H}(q,p,t) = \frac{1}{2}(q^T Q_1 q + q^T Q_2 p + p^T Q_3 q + p^T Q_4 p)\,.
\end{equation}
We can write this compactly as
\begin{equation}\label{eq:Ham}
    \mathcal{H}(z,t) = \frac{1}{2}z^T  Q z\,,
\end{equation}
where 
\begin{equation}
    Q=\begin{pmatrix}
        Q_1&Q_2\\Q_3&Q_4
    \end{pmatrix}\,,
    \hspace{0.1in}
    z=\begin{pmatrix}
        q\\p
    \end{pmatrix}\,.
\end{equation}
We can then exploit the rich symplectic structure of the phase space of such systems. Below, we give some basic information about this structure that we will use in our algorithms later.

For any integer $n$, a symplectic matrix $S$ is a $2n\times 2n$ matrix (over the reals) such that 
\begin{equation}
    S^T J S =J\,,
\end{equation}
where $S^T$ is the transpose of $S$ and $J$ is defined as the following block matrix
\begin{equation}
    J = \begin{pmatrix}
        0 & I\\-I & 0
    \end{pmatrix}\,.
\end{equation}
The set of all symplectic matrices forms a group under multiplication denoted Sp$(2n, \mathbb{R})$. This continuous group also has the structure of a Lie group whose Lie algebra is given by
\begin{equation}
    \mathfrak{sp}(2n,\mathbb{R}) = \{X\in M_{2n\times 2n}:JX + X^TJ=0\}\,.
\end{equation}
In terms of block matrices, the symplectic algebra can be written as 
\begin{equation}
    X=\begin{pmatrix}
        A&B\\C&D
    \end{pmatrix}\,,
\end{equation}
such that 
\begin{align}
    &A=-D^T\\
    &C=C^T\\
    &B=B^T\,.
\end{align}
The exponential map takes the lie algebra to the Lie group of symplectic matrices i.e., $\exp(X)=S$, where $X\in\mathfrak{sp}(2n,\R)$ and $S\in$Sp$(2n,\R)$. Without any additional constraints on $X$, the exponential map is not a diffeomorphism, which means that one cannot exponentiate to get to an arbitrary symplectic matrix. In terms of these quantities, we can write Hamilton's equation \eq{Ham_eqs} as follows
\begin{equation}\label{eq:dz_dt}
    \frac{dz}{dt} = JQz\,,
\end{equation}
where $Q$ defines the quadratic Hamiltonian in \eq{Ham}. The solution to this
\begin{equation}\label{eq:z_soln}
    z(t)=\exp(JQt)z(0)\,,
\end{equation}
where $z(0)$ is the initial condition at $t=0$.

When $Q$ is a positive definite matrix, Williamson's theorem \cite{de2006symplectic} implies that $JQ$ is diagonalizable (not necessarily unitarily diagonalizable) and its eigenvalues are all purely imaginary. Loosely speaking, $J$ plays the role of the complex unit $i=\sqrt{-1}$. When $Q$ is a positive definite matrix, we can convert the evolution in \eq{z_soln} to a quantum Hamiltonian evolution as follows. First, define $H=i\sqrt{Q}J\sqrt{Q}$. It is easy to see that when $Q$ is positive, $H$ is Hermitian (using the fact that $J^\dag=-J$). Now define
\begin{equation}
    y(t)=\sqrt{Q}z(t)\,.
\end{equation}
Then \eq{z_soln} becomes
\begin{equation}
    y(t) = \sqrt{Q}\exp(JQt)(\sqrt{Q})^{-1}y(0) = \exp(\sqrt{Q}J\sqrt{Q}t)y(0) = \exp(iHt)y(0)\,.
\end{equation}
This shows that $y(t)$ evolves according to the quantum Hamiltonian $H$. However, $y(t)$ may not always be easy to deal with from the point of view of extracting useful information in a quantum algorithm. At times, we might be interested in obtaining the solution
\begin{equation}
    x = \begin{pmatrix}
        q\\\dot{q}
    \end{pmatrix}\,.
\end{equation}
It can be seen that
\begin{equation}
    z=\begin{pmatrix}
        I&0\\0&M
    \end{pmatrix}
    x\,,
\end{equation}
where recall that 
\begin{equation}
    z=\begin{pmatrix}
        q\\p
    \end{pmatrix}\,.
\end{equation}
This means that from \eq{dz_dt}
\begin{equation}
   \begin{pmatrix}
        I&0\\0&M
    \end{pmatrix} \frac{dx}{dt} = 
    JQ\begin{pmatrix}
        I&0\\0&M
    \end{pmatrix}x\,,
\end{equation}
This gives us
\begin{equation}
    \frac{dx}{dt} = \begin{pmatrix}
        I&0\\0&M^{-1}
    \end{pmatrix} JQ \begin{pmatrix}
        I&0\\0&M
    \end{pmatrix} x\,.
\end{equation}
When $Q$ has a special block-diagonal form below (such as in the case of coupled oscillators)
\begin{equation}
    Q=\begin{pmatrix}
        V&0\\0&M^{-1}
    \end{pmatrix}\,,
\end{equation}
we have
\begin{equation}\label{eq:dx_dt}
    \frac{dx}{dt} = \begin{pmatrix}
        0&I\\-M^{-1}V&0
    \end{pmatrix} x\,.
\end{equation}

\subsection{Quadratic Hamiltonians with dissipation}
Non-conservative systems more accurately model real-world mechanical systems. In this subsection, we consider dissipation, which could arise due to drag, damping or other resistive forces on conservative systems. In such systems, one defines the Lagrangian $L$ in the usual way without including dissipation and then adds terms such as the Rayleigh dissipation function \cite{goldstein2002classical}. However, the Rayleigh dissipation function does not account for forcing or source terms, which are added separately. All of this modifies the first variation as follows.
\begin{equation}
    \frac{d}{dt}\frac{\partial L}{\partial\dot{q}_i} - \frac{\partial L}{\partial q^i} = r_i\,,
\end{equation}
where $r_i$ is models the role of dissipation and sources. The equation of motion can be written as
\begin{equation}
    M\ddot{q} + R\dot{q} + Vq + f=0\,.
\end{equation}
In the above equation, $M$ comes from the kinetic energy term and corresponds the mass matrix, $R$ is a matrix capturing damping or dissipation, $V$ is the potential energy matrix and $f$ corresponds to sources or forcing terms. Several interesting systems such as damped coupled oscillators \cite{goldstein2002classical}, dynamics of rotating machinery \cite{friswell2010dynamics} and gyroscopic systems \cite{LANCASTER2013686} have equations of motion of this type. To convert this system of second order differential equations to Hamilton's first order equations, consider a Lagrangian as follows.
\begin{equation}
    L = \frac{1}{2}\dot{q}^TM\dot{q} - \frac{1}{2}q^TVq\,.
\end{equation}
We can define the Hamiltonian in the same way as before (via a Legendre transformation) as 
\begin{equation}
    H = p\cdot \dot{q} - L\,,
\end{equation}
where
\begin{equation}
    p_i = \frac{\partial L}{\partial\dot{q}_i} = M_i\dot{q}_i\,.
\end{equation}
Now one can define the following modification of Hamilton's canonical equations \cite{goldstein2002classical}.
\begin{align}
    &\dot{q}_j = \frac{\partial H}{\partial p_j}\,, \\
    &\dot{p}_j = -\frac{\partial H}{\partial q_j} + r_j\,,
\end{align}
where $r_j$ captures the dissipation and source terms and comes from the matrix $R$. In our case, these are the terms involving $R$ and $f$. This gives us
\begin{equation}
    r_j = -\sum_i R_{j i}\dot{q}_i - f_j\,.
\end{equation}
This becomes
\begin{equation}\label{eq:z_basis}
    \frac{dz}{dt} = \begin{pmatrix}
        0 & M^{-1}\\-V & -RM^{-1}
    \end{pmatrix}z -\tilde{f}\,,
\end{equation}
where
\begin{equation}
    z = \begin{pmatrix}
        q\\p
    \end{pmatrix}\,,\hspace{0.1in} \text{and}\hspace{0.1in} \tilde{f} = \begin{pmatrix}
        0\\f
    \end{pmatrix}\,.
\end{equation}
Next, we consider $y=\sqrt{Q}z$. In this basis, we get
\begin{equation}
    \frac{dy}{dt} = \sqrt{Q}\begin{pmatrix}
        0 & M^{-1}\\-V & -RM^{-1}
    \end{pmatrix}\sqrt{Q^{-1}}y - \sqrt{Q}\tilde{f}\,,
\end{equation}
Let $b=-\sqrt{Q}\tilde{f}$. Using
\begin{equation}
    Q=\begin{pmatrix}
        V&0\\0&M^{-1}
    \end{pmatrix}\,,
\end{equation}
we get
\begin{equation}\label{eq:y_basis}
    \frac{dy}{dt} = \begin{pmatrix}
        0 & \sqrt{VM^{-1}}\\
        -\sqrt{M^{-1}V} & -\sqrt{M^{-1}}R\sqrt{M^{-1}}
    \end{pmatrix}y+b\,.
\end{equation}
Finally, switching to $x=(q,\dot{q})^T$, we get
\begin{equation}\label{eq:x_basis}
    \frac{dx}{dt} = \begin{pmatrix}
        0 & I\\-M^{-1}V & -M^{-1}R
    \end{pmatrix}x + b\,,
\end{equation}
where (it turns out due to the form of $\tilde{f}$)
\begin{equation}
    b = -\begin{pmatrix}
        I & 0 \\0& M^{-1}
    \end{pmatrix} \tilde{f}\,.
\end{equation}
This can be written as 
\begin{equation}\label{eq:Ham_eq_diss}
    \frac{dx}{dt}=Ax+b\,,
\end{equation}
where 
\begin{equation}
    A = \begin{pmatrix}
        0 & I\\-M^{-1}V & -M^{-1}R
    \end{pmatrix}\,.
\end{equation}
and $b$ is the transformed source term.

\subsection{Optimal control}\label{sec:optimal}
In this subsection, we continue with calculus of variations and discuss applications to optimal control. In this setting, a \emph{control} parameter is included in the Lagrangian which controls the position variable through a differential equation. The goal is to find the value of the control parameter as a function of time such that a cost functional is minimized. Such problems have many practical applications in operations research \cite{liberzon2012calculus}.

In mathematical terms, the problem of optimal control can be described as follows. Let $q(t)$ be the position variable as before and $u(t)$ be the control parameter. Suppose that the two are related by an ODE as follows.
\begin{equation}
    \dot{q} = F(q,u,t)\,,
\end{equation}
where $F$ is some function (which may be nonlinear). Later, we will relax that condition and in fact, assume a linear relationship. For the general problem definition however, we will assume $F$ to be arbitrary. The problem is to find the optimal $u(t)$ such that the total cost $J(u)$ is minimized, where $J(u)$ is defined as
\begin{equation}
    J(u) = \int_{t_i}^{t_f} L(q,\dot{q},t) dt + K(t_f,q_f)\,,
\end{equation}
Here $L$ is the Lagrangian (also the running cost) and  $K(t_f,q_f)$ is the terminal cost function that depends on the terminal time $t_f$ and the terminal position $q_f$. Here the final state $q_f$ and the final time $t_f$ are either fixed or free. This cost functional is sometimes called the Bolza functional \cite{liberzon2012calculus}. One can also convert this to a problem of maximization by taking the negative of the Lagrangian.

While it is not easy to solve the above problem for arbitrary functions and cost functionals, there are useful necessary and sufficient conditions that can be used to check if a given curve is optimal. Often, these conditions narrow down the search space making it possible to find a local minimum. A local minimum is defined as the set of all $u^\ast$ such that $J(u^\ast)\leq J(u)$ for all $|u^\ast - u|\leq \epsilon$.

There are three main necessary conditions for any curve $u(t)$ to be a minimum. The first of these conditions is that the first variation of the Lagrangian must equal zero leading to the Euler-Lagrange equations described above. The second is called the Legendre condition and it states that the second derivative of the Lagrangian with respect to $\dot{q}$ is positive. More precisely, define the matrix $L_{\dot{q}\dot{q}}$ as
\begin{equation}
    [L_{\dot{q}\dot{q}}]_{i,j} = \frac{\partial^2\mathcal{H}}{\partial \dot{q}_i\partial \dot{q}_j}\,,
\end{equation}
then the Legendre condition states that
\begin{equation}
    L_{\dot{q}\dot{q}}\geq 0\,.
\end{equation}
Finally, the third necessary condition is called Weierstrass condition and it states that the Lagrangian must be locally convex. To state it in mathematical form, define the so called Weierstrass excess function as follows.
\begin{equation}
    E(q,v,w,t) =L(q,v,t) - L(q,w,t) - (w-v)L_v(q,v,t)\,,
\end{equation}
where $v=\dot{q}$. Then the necessary condition states that $E(q,v,w,t)\geq 0$ for all $t\in [t_i,t_f]$ and for all $w$.

Sufficient conditions are similar to the above necessary conditions but stronger. The following set of conditions is sufficient for a curve $q(t)$ to be the (strong local) minimum.
\begin{enumerate}
    \item $q(t)$ and $\dot{q}(t)$ satisfy the Euler-Lagrange equations (or equivalently, $q$ and $p$ satisfy Hamilton's canonical equations),
    \item $L_{\dot{q}\dot{q}}>0$,
    \item $E(x,v,w,t)\geq 0$, for all $w$ and where $|x-q|\leq \epsilon$ and $|v-\dot{q}|\leq \epsilon$.
    \item There are no conjugate points in the interval $[t_i,t_f]$.
\end{enumerate}
The first condition is just the first variation and the second comes from taking the second variation of the Lagrangian. As pointed out above, these two conditions are similar to the vanishing of the first derivative and the positivity of the second derivative as sufficient conditions for the existence of a minimum of a smooth function. The third condition is same as before. Here the last condition is sometimes called the Jacobi condition, which we explain next since this is relevant to our algorithm for the Riccati equation.

Conjugate points are points on a manifold which need not minimize geodesics. The Jacobi equation can be used to determine if there exist conjugate points in any interval. Applied to the above Lagrangian, the definition states that $t_0$ and $t_1$ are conjugate if there exists a function $h(t)$ satisfying
\begin{equation}
    \frac{d}{dt}(L_{\dot{q}\dot{q}}\dot{h} + L_{\dot{q}q}h) -L_{q\dot{q}}\dot{h} - L_{qq}h = 0\,,
\end{equation}
such that $h(t_0)=0$ and $h(t_1)=0$. It turns out that the non-existence of conjugate points is related to the existence of a solution of the Riccati equation (which we explain in \sec{conj}).

The two main approaches to solving problems in optimal control are the Pontryagin principle approach (also called the maximum principle) and the Hamilton-Jacobi-Bellman approach. The approach using the Pontryagin principle uses Hamilton's canonical equations along with inequality conditions, where as the Hamilton-Jacobi-Bellman equation is a partial differential equation. These are described in detail in \cite{liberzon2012calculus} for instance. Here, we focus on a special case called the linear quadratic regulator (LQR) problem, where the cost function or Lagrangian is quadratic and the ODE that describes the relationship between the position $q(t)$ and the control $u(t)$ is linear. We give a quantum algorithm for this problem by first solving a closely related differential equation called the Riccati equation. Since our quantum algorithm to solve the LQR problem uses the Riccati equation, we describe this equation next. The Riccati equation is intimately related to the optimal control problem as explained in \cite{clarke_zeidan_1986}. 

\subsection{Riccati equation and flows on the Grassmannian}\label{sec:Grass}
The Grassmannian (denoted Gr$(N,M)$) is the set of $M$ dimensional vector spaces of an $N$ dimensional vector space (over the real or complex numbers). While the field for these vector spaces can be arbitrary such as the quaternions, we restrict attention to the complex numbers in this paper. The Grassmannian has the structure of a smooth manifold since it can be viewed as the homogeneous space GL$(N)/($GL$(N-M)\times $GL$(M))$. The action of the general linear group on the Grassmannian is transitive and gives rise to a vector bundle where the base manifold is the Grassmannian and the fibers are isomorphic to GL$(M)$. This vector bundle is called the Stiefel manifold $\mathbb{V}(N,M)$. It can be viewed as the set of $M$ frames i.e., the set of tuples of $M$ linearly independent vectors in $\C^N$. This set of points constitute the non-compact Stiefel manifold. The compact version of this consists of the set of \emph{orthonormal} $M$ frames.

Any point in $\mathbb{V}(N,M)$ can be represented by an $N\times M$ matrix $y$ consisting of linearly independent columns. Two such points have the same span if they differ by an element of the general linear group. This means that the projection map $\pi:\mathbb{V}(N,M)\rightarrow$Gr$(N,M)$ takes $y$ to the vector space consisting of the span of the columns of $y$. One can think of the matrix $y$ as representing a point on the Grassmannian up to multiplication by an element of GL$(M)$ i.e., $y$ and $\tilde{y}$ represent the same point on Gr$(N,M)$ if $\tilde{y} = yg$, where $g\in$GL$(M)$. Strictly speaking, the point $y$ lies in the Stiefel manifold and points in GL$(M)$ is on the fiber above the point $\pi(y)$ on the Grassmannian.

In terms of their action on vectors $z$ of complex numbers in $\C^N$, the coordinates $y$ can be defined as
\begin{equation}\label{eq:z_12}
    z_1 = y\cdot z_2\,,
\end{equation}
where $z_1$ is the set of the first $N-M$ components of $z$ and $z_2$ is the rest of the $M$ components. Choosing these coordinates $y$ is equivalent to choosing a set of $M$ independent vectors in $\mathbb{C}^N$. To view the Riccati equation as a flow on $Gr(M,N)$, first consider a vector field on the Steifel manifold given by
\begin{equation}
    z' =  A z\,,
\end{equation}
where $A$ is some $N\times N$ matrix and $z'$ is the derivative of $z$ (with respect to some parameter $t$). Let $A$ have the block form
\begin{equation}\label{eq:A_mx}
    A=\begin{pmatrix}
        F_1 & F_0\\F_2&F_3
    \end{pmatrix}\,,
\end{equation}
where $F_0$ is $N-M\times M$, 
Then we have
\begin{align}\label{eq:z_12_matrix}
    z_1' = F_1 z_1 + F_0z_2\,,\\
    z_2' = F_2z_1 + F_3z_2\,.
\end{align}
Taking the derivative of \eq{z_12} and using \eq{z_12_matrix}, we get
\begin{equation}
    F_1z_1 + F_0z_2 = z_1' = y'z_2 + yz_2' = y'z_2 + y(F_2z_1 + F_3z_2)\,.
\end{equation}
Using \eq{z_12} again, we get
\begin{equation}
    y'z_2 = (F_1y + F_0 - yF_2y - yF_3)z_2\,.
\end{equation}
Since $z_2$ is arbitrary, we obtain the Riccati equation as a flow defined on the Grassmannian. With this description of the Riccati equation, it becomes easy to see a natural linearization method which we describe next.

\subsection{Riccati equation and the M\"{o}bius transformation}\label{sec:Mobius}
Here, we introduce a well-known transformation called the M\"{o}bius transformation which linearizes the Riccati equation. This is a slight generalization of the one derived in \cite{SS99} and \cite{HLAVATY1984246}. This version is more convenient for our quantum algorithm since it allows for a constant term.
\begin{proposition}\label{prop:Riccati_ODE}
Define the following initial value problem for a linear differential equation.
\begin{equation}\label{eq:linear_Riccati}
    \dot{x}=Ax+b\,,
\end{equation}
where 
\begin{align}
    x(t)=\begin{pmatrix} u(t) \\v(t)\end{pmatrix} \,,\hspace{0.3in}\label{eq:u_v}
    A = \begin{pmatrix} F_1 & F_0\\F_2 & F_3\end{pmatrix}\,,\hspace{0.3in}
    b =\begin{pmatrix} e \\f\end{pmatrix}\,,
\end{align}
and where $u$ and $e$ are $N\times M$ matrices, $v$ and $f$ are $M\times M$ matrices. The initial conditions are $u(0)=y_0$ and $v(0)=I$. Then whenever $v(t)^{-1}$ is defined, $y(t)=(u(t)+w)v(t)^{-1}$ is a solution to the matrix Riccati equation \eq{Riccati_ivp}, where $w$ can be picked to be any time-independent vector and $e$ and $f$ are defined as
\begin{equation}\label{eq:e_f}
    e=F_1w\,\text{ and }\,f=F_2w\,.
\end{equation}
\end{proposition}
\begin{proof}
From the definition of $y$ as
\begin{equation}
    y = (u+w)v^{-1}\,,
\end{equation}
we see that taking the derivative of $y$ gives us
\begin{equation}\label{eq:y_dot}
    \dot{y}=\dot{u}v^{-1} - (u+w)v^{-1}\dot{v}v^{-1}\,.
\end{equation}
From \eq{linear_Riccati} and \eq{u_v}, we have
\begin{align}
    &\dot{u} = F_1u+F_0v+e\\
    &\dot{v} = F_2u+F_3v+f\,.
\end{align}
Replacing $\dot{u}$ and $\dot{v}$ in \eq{y_dot} above, we get
\begin{align}
    \dot{y}&=(F_1u+F_0v +e)v^{-1} + (u+w)v^{-1}(F_2u+F_3v+f)v^{-1}\\
    &= F_0+ F_1y + yF_2y +yF_3\,,
\end{align}
where we used the definitions of $e$ and $f$ from \eq{e_f}. The last equation is the Riccati equation with the initial condition $y(0)=u(0)v(0)^{-1}=y_0$.
\end{proof}

When using the M\"{o}bius transformation to solve the Riccati equation, one can run into singularities i.e., places where $v(t)$ is not invertible. This does not mean that the solution to the Riccati equation does not exist since $y(t)$ is always defined. In \cite{SS99}, it was shown that this is due to the solution crossing coordinate patches in the Grassmannian manifold. If one considers the solution on the single patch, then there is no singularity in the solution. It turns out that in our applications, the quantity $v(t)$ is invertible and in fact, the condition number can be assumed to be well-behaved. As we explain next, for the problem of finding the solution of an optimal control problem, the fact that the Jacobi condition is satisfied is equivalent to $v$ in \eq{u_v} being non-singular. 

\subsection{Conjugate points, Jacobi's condition and the Riccati equation}\label{sec:conj}
In this subsection, we describe the relationship between conjugate points and the Riccati equation. Given a Lagrangian $L$, conjugate points are defined as follows.
\begin{definition}
A point $t=c$ is called a conjugate point in the interval $(a,b)$ if there exists a non-trivial solution to the following boundary value problem in $(a,c)$.
\begin{align}
    &\frac{d}{dt}[L_{\dot{q}\dot{q}}\dot{h} +L_{\dot{q}q}h ] - L_{q\dot{q}}\dot{h} - L_{qq}h = 0\label{eq:Jacobi}\\
    &h(a)=h(c)=0\,,
\end{align}
where $L(q,\dot{q},t)$ is the Lagrangian and terms such as $L_{\dot{q}q}$ are its partial derivatives.
\end{definition}
The Jacobi condition for an interval $(a,b)$ states that there are no conjugate points in that interval. It turns out that there exists a Riccati equation corresponding to the Jacobi equation \cite{Reid}. It can be defined as follows (via its linearization)
\begin{align}\label{eq:Riccati_Jacobi}
    &\dot{U}(t)=AU(t) + BV(t)\\
    &\dot{V}(t)=CU(t) - A^\dag V(t)\,,
\end{align}
where $A, B$ and $C$ are defined as
\begin{align}
    &A = -L_{\dot{q}\dot{q}}^{-1}L_{\dot{q}q}\\
    &B = L_{\dot{q}\dot{q}}^{-1}\\
    &C = L_{qq}-L_{q\dot{q}}L_{\dot{q}\dot{q}}^{-1}L_{\dot{q}q}\,.
\end{align}
The above is the matrix Riccati equation. Define the corresponding vector Riccati equation
\begin{align}\label{eq:Riccati_Jacobi_vec}
    &\dot{u}(t)=Au(t) + Bv(t)\\
    &\dot{v}(t)=Cu(t) - A^\dag v(t)\,,
\end{align}
where $u$ and $v$ are vectors. By defining 
\begin{align}
    &u = h\\
    &v = L_{\dot{q}\dot{q}}\dot{h} +L_{\dot{q}q}h\,,
\end{align}
the Jacobi equation \eq{Jacobi} is equivalent to the set of equations \eq{Riccati_Jacobi}. To see this, notice that the Jacobi equation is
\begin{equation}
    \dot{v} - L_{q\dot{q}}\dot{u} - L_{qq}u = 0\,.
\end{equation}
Now Riccati equation can be written as 
\begin{align}
    &v = B^{-1}(\dot{u} - Au)\\
    &\dot{v} = Cu - A^\dag v\,,
\end{align}
which leads to
\begin{equation}
    \dot{v} = (C - A^\dag B^{-1}A)u - A^\dag B^{-1}\dot{u}\,.
\end{equation}
Now using the definitions of $A, B$ and $C$, we have
\begin{align}
    &C - A^\dag B^{-1}A = L_{qq}\\
    &A^\dag B^{-1} = L_{q\dot{q}}\,,
\end{align}
which shows that the Jacobi equation is equivalent to the Riccati equation. The following lemma was proved in \cite{clarke_zeidan_1986}.

\begin{lemma}\label{lem:conjugate}
If the Jacobi condition is satisfied i.e., there are no conjugate points, then $V(t)$ in the Riccati equation \eq{Riccati_Jacobi} is invertible.
\end{lemma}
In this paper, we assume a quantitative version of the Jacobi condition where $V(t)$ is invertible and has a condition number bounded by $\kappa_V$.

\subsection{From the Riccati equation to the Hamilton-Jacobi-Bellman equation}\label{sec:HJB}
The Hamilton-Jacobi-Bellman equation is a nonlinear partial differential equation of the form below.
\begin{equation}\label{eq:HJB}
    \frac{\partial V^\ast (x(t),t)}{\partial t}=-\min_{u(t)}\Big(C(x(t),u(t),t)+ \Big[\frac{\partial V^\ast (x(t),t)}{\partial t}\Big]^\dag f(x(t),u(t),t)\Big)\,.
\end{equation}
Here $C(x(t),u(t),t)$ is a cost function, which we assume is a quadratic function of the following form.
\begin{equation}
    C(x(t),u(t),t)=u^\dag(t)R(t)u(t) + x^\dag(t)Q(t)x(t)\,,
\end{equation}
and that $V^\ast$ is of the form
\begin{equation}
    V^\ast (x(t),t)=x^\dag(t)P(t)x(t)\,,
\end{equation}
and $f(x,u,t)$ has a linear form i.e.,
\begin{equation}
    f(x(t),u(t),t)=F(t)x(t)+G(t)u(t)\,.
\end{equation}
Using this and re-writing the HJB equation, we get
\begin{equation}\label{eq:HJB_special}
    x^\dag \dot{P}x(t)=-\min_{u(t)}[u^\dag R u+x^\dag Qx + 2x^\dag PFx + 2x^\dag PG u]\,.
\end{equation}
The next result can be found in \cite{anderson2007optimal}.
\begin{lemma}\label{lem:Riccati_HJB}
Using the quadratic form for the cost function above and the linear form for $f$, the solution to the HJB equation \eq{HJB} can be obtained from a matrix Riccati equation.
\end{lemma}
\begin{proof}
Using the identity
\begin{align}
    u^\dag Ru + x^\dag Qx + 2x^\dag PFx + 2x^\dag PGu = &(u+R^{-1}G^\dag Px)^\dag R(u+R^{-1}G^\dag Px) +\nonumber\\
    & x^\dag (Q-PGR^{-1}G^\dag P+PF+F^\dag P)x\,.
\end{align}
Since $R$ is positive definite and the second term above does not depend on $u$, the minimum of \eq{HJB_special} is obtained at 
\begin{equation}
    \bar{u}(t)=-R^{-1}G^\dag Px\,.
\end{equation}
Using this in \eq{HJB_special}, we get
\begin{equation}
    x^\dag \dot{P}x(t)=-x^\dag (Q-PGR^{-1}G^\dag P+PF+F^\dag P)x\,.
\end{equation}
Since this holds for all $x(t)$, we have
\begin{equation}\label{eq:mre-hjb}
    \dot{P}=PGR^{-1}G^\dag P - PF - F^\dag P -Q\,.
\end{equation}
This is a matrix Riccati equation. Using its solution $P(t)$, we can get the solution of the HJB equation by setting
\begin{equation}\label{eq:mjb-soln}
    \bar{u}(t)=-R^{-1}G^\dag Px(t)\,.
\end{equation}
\end{proof}
In applications to control problems such as linear quadratic control, one needs to set a boundary condition rather than an initial condition i.e., $P(t_f)=P_f$ and obtain a solution $P(t)$ for all time $t\leq t_f$.

\section{Computational preliminaries and problem statements}\label{sec:comp_prelims}
\subsection{Block encoding}\label{sec:block}

We first define a useful concept called block-encoding (see \cite{10.1145/3313276.3316366} for more details) and cite some results from the literature that we use later in the paper.
\begin{definition}\label{defn:block_encoding}
Suppose $A$ is an $a$-qubit operator, $\alpha,\epsilon\in \mathbb{R}_+$ and $b\in \mathbb{N}$, then an $a+b$ qubit unitary $U$ is said to be an $(\alpha, b, \epsilon)$ block encoding of $A$ if
\begin{equation}
    \|A-\alpha(\bra{0}^b\otimes I)U(\ket{0}^b\otimes I)\|\leq \epsilon\,.
\end{equation}
\end{definition}
\begin{definition}\label{defn:Hermitian_complement}
We denote by $\bar{A}$ the Hermitian complement of $A$ defined below 
\begin{equation}
    \bar{A}=\begin{pmatrix}0&A\\A^\dag &0
    \end{pmatrix}\,.
\end{equation}
\end{definition}
Next, we state the relevant results such as block-encoding a sparse matrix, adding, multiplying and inverting block-encoded matrices.
\begin{lemma}[\cite{10.1145/3313276.3316366}, Lemma 48]\label{lem:sparse_matrix}
Suppose $A$ is an $s_r$-row-sparse and $s_c$-column-sparse $a$-qubit matrix where each element of $A$ has absolute value at most 1 and where we have oracle access to the positions and values of the nonzero entries, then we have an implementation of $(\sqrt{s_rs_c},a+3,\epsilon)$ block encoding of $A$ with a single use of $O_r$ and $O_c$ and two uses of $O_A$ and using $O(a+\log^{2.5}(s_rs_c/\epsilon))$ elementary gates and $O(b+\log^{2.5}(s_rs_c/\epsilon))$ ancillas.
\end{lemma}
The next result is from \cite{chakraborty_et_al:LIPIcs:2019:10609} and gives the complexity of inverting a matrix.
\begin{lemma}[\cite{chakraborty_et_al:LIPIcs:2019:10609}, Lemma 9]\label{lem:matrix_inversion}
Let $A$ be a matrix with condition number $\kappa\geq 2$. Let $H$ be the Hermitian complement of $A$ and suppose that $I/\kappa\leq H\leq I$. Let
\begin{equation}
    \delta = o(\epsilon/\kappa^2\log^3(\kappa^2/\epsilon))\,.
\end{equation}
If $U$ is an $(\alpha, a, \delta)$ block encoding of $H$ that has gate complexity $T_U$, then we can implement a 
\begin{equation}
    (2\kappa,a+O(\log(\kappa^2\log(1/\epsilon))),\epsilon)
\end{equation}
block encoding of $H^{-1}$ with gate complexity
\begin{equation}
    O(\alpha\kappa(a+T_U)\log^2(\kappa^2/\epsilon))\,.
\end{equation}
\end{lemma}
\begin{remark}
In the above lemma, when $H$ is a Hermitian complement of a matrix $A$, then after inversion, the resulting block encoded matrix satisfies
\begin{equation}
    \|A^{-1}-2\kappa(\bra{0}^b\otimes \bra{1}\otimes I)U(\ket{0}^b\otimes\ket{0}\otimes I)\|\leq \epsilon\,,
\end{equation}
where the middle qubit corresponds to the qubit needed for the Hermitian complement. By appending an $X$ gate to the block encoding for that qubit, this can be brought into the above form of block encoding of $A^{-1}$ (as in \defn{block_encoding}).
\end{remark}
The next result gives the complexity of performing matrix arithmetic on block-encoded matrices.
\begin{lemma}[\cite{chakraborty_et_al:LIPIcs:2019:10609}, Lemma 6, \cite{10.1145/3313276.3316366}, Lemmas 52, 54]\label{lem:matrix_arithmetics}
If $A$ has an $(\alpha,a,\epsilon)$ block encoding with gate complexity $T_A$ and $B$ has a $(\beta,b,\delta)$ block encoding with gate complexity $T_B$, then
\begin{enumerate}
    \item $\bar{A}$ has an $(\alpha,a+1,\epsilon)$ block encoding that can be implemented with gate complexity $O(T_A)$.
    \item $A+B$ has an $(\alpha+\beta,a+b,\beta\epsilon+\alpha\delta)$ block encoding that can be implemented with gate complexity $O(T_A+T_B)$.
    \item $AB$ has an $(\alpha\beta,a+b,\alpha\delta+\beta\epsilon)$ block encoding with gate complexity $O(T_A+T_B)$.
\end{enumerate}
\end{lemma}

Next, we need the following theorem to implement the QLSA.
\begin{theorem}[\cite{QLSA_linear_kappa}, Theorem 19]\label{thm:block_QLSA}
Let $A$ be such that $\|A\|=1$ and $\|A^{-1}\|=\kappa$. Given a oracle block encoding of $A$ and an oracle for implementing $\ket{b}$, there exists a quantum algorithm which produces the normalized state $A^{-1}\ket{b}$ to within an error $\epsilon$ using 
\begin{equation}
    O(\kappa\log 1/\epsilon)
\end{equation}
calls to the oracles.
\end{theorem}
\begin{remark}
In using \lem{matrix_inversion} and \thm{block_QLSA}, we will apply them to matrices with norms that are $O(1)$ rather than $1$. This does not change the asymptotic scaling of the gate complexities.
\end{remark}
The next theorem from \cite{chakraborty_et_al:LIPIcs:2019:10609} gives a quantum algorithm to estimate the norm of the solution state.
\begin{theorem}\label{thm:norm_estimate}
Let $Ax=b$ be an $N\times N$ linear system with an invertible matrix $A$ with sparsity $s$ and condition number $\kappa$. Given an oracle that computes s normalized version of $b$, there is a quantum algorithm that outputs a quantity $\tilde{x}$ such that
\begin{equation}
    \Big|\tilde{x} - \|A^{-1}b\|\Big|\leq \gamma \|A^{-1}b\|\,,
\end{equation}
with constant probability and with gate complexity 
\begin{equation}
    O(\frac{\kappa}{\gamma}(sT_U\log^2\frac{\kappa}{\gamma} + T_b)\log^3\kappa\log\log\kappa)\,,
\end{equation}
where 
\begin{equation}
    T_U = O(\log N + \log^{2.5}(\frac{s\kappa}{\gamma}\log \frac{\kappa}{\gamma}))\,.
\end{equation}
\end{theorem}

Next we define the oracle access framework that is commonly assumed in quantum algorithms in the literature. Suppose the matrix $A$ in linear ODEs such as \eq{Ham_eq_diss} is a block $2\times 2$ matrix composed of blocks $F_0$ through $F_3$, we assume that the entries of these matrices are accessed through the following oracles.
\begin{definition}\label{defn:oracles}
We assume that each $F_m$ can be accessed through the following oracles.
\begin{align}
    &O_r\ket{i,k}=\ket{i,r_{ik}}\,,\\
    &O_c\ket{i,k}=\ket{i,c_{ik}}\,,\\
    &O_a\ket{i,j}\ket{0}^{\otimes b}=\ket{i,j,F_m(i,j)}\,,
\end{align}
where $r_{ij}$ (respectively $c_{ji}$) is the $j^{th}$ nonzero entry in the $i^{th}$ row (resp. column). If there are fewer than $j$ entries, then $r_{ij}$ (resp. $c_{ji}$) is set to $j+2^n$. Here $F_m(i,j)$ is a binary representation of the $(i,j)$ matrix entry of $F_m$. Here $F_m$ is one of $F_0$, $F_1$, $F_2$ or $F_3$.

Finally, we define an oracle $O_{init}$ to prepare a normalized version of the initial state $z_0$. We assume that the norm of $z_0$ is known. Specifically, let $O_{init}$ be any unitary that maps $\ket{1}\ket{\phi}$ to $\ket{1}\ket{\phi}$ for any state $\ket{\phi}$ and $\ket{0}\ket{0}$ to $\ket{0}\ket{\bar{z}_{0}}$, where $\bar{z}_{0}=z_{0}/\|z_{0}\|$.
\end{definition}

\begin{remark}
One interesting and useful aspect of linearization of the Riccati equation using a M\"{o}bius transformation is that we do not need to be able to prepare a quantum state proportional to $F_0$ as was the case with Carleman linearization \cite{Liue2026805118}. We only need to access its entries via an oracle, which is easier.
\end{remark}

\subsection{Problem statements}\label{sec:probs}
In this subsection, we define precisely the problems for which we have efficient algorithms. The first problem is to simulate the evolution of a classical physical system with a quadratic Hamiltonian with dissipation and forcing terms.
\begin{problem}\label{prb:class_mech}
Given oracle access to matrices in the following equation of motion
\begin{equation}\label{eq:M_R_V_b_eq}
    M\ddot{q} + R\dot{q} + Vq +f = 0\,,
\end{equation}
and an oracle to prepare the normalized versions of the initial state $q(0)$ and $f$. We are also given estimates of their norms and promised that the damping matrix $R$ is positive. The problem is to 
\begin{enumerate}
    \item produce the normalized quantum state proportional to $(q(T),p(T))^T$ at some time $T$, and
    \item output an estimate of the kinetic energy at time $T$ to within an additive error $\epsilon$.
\end{enumerate}
\end{problem}
In optimal control, we are first interested in the following initial value problem of the Riccati equation.
\begin{problem}\label{prb:Riccati_vector}
Consider the following initial value problem for the vector Riccati equation
\begin{equation}\label{eq:Riccati_ivp}
    \dot{y}=F_0+F_1y+yF_2y+yF_3\,,\hspace{1cm} y(0)=y_0\,,
\end{equation}
where $y$ is an $N$ dimensional vector with initial value $y_0$ and the coefficients $F_0$, $F_1$, $F_2$ and $F_3$ are time-independent. Given oracle access to the entries of $F_i$, the problem is to output a quantum state proportional to $y(T)$ for some $T$ to within a given $\ell_2$-norm error $\epsilon$.
\end{problem}
The next problem we consider is to find the solution of the matrix Riccati equation as a block-encoded unitary.
\begin{problem}\label{prb:Riccati_matrix}
Given the matrix Riccati equation 
\begin{equation}
    \dot{y}=F_0+F_1y+yF_2y+yF_3\,,\hspace{1cm} y(0)=y_0\,,
\end{equation}
where $y$ is an $N\times M$ dimensional matrix with initial value matrix $y_0$ and the coefficients $F_0$, $F_1$, $F_2$ and $F_3$ are time-independent. Given oracle access to the entries of $F_i$ and $y_0$, the problem is to output a block encoded unitary proportional to $y(T)$ for a given $T$ to within a given $\ell_2$-norm error $\epsilon$.
\end{problem}

\subsection{Norm of the exponential}\label{sec:C(A)}
When implementing quantum algorithms to solve linear ODEs of the form
\begin{equation}
    \dot{z} = Az + b\,,
\end{equation} 
It was shown in \cite{krovi2022improved} that the gate complexity depends on the norm of the exponential of the matrix $A$ i.e., it depends on
\begin{equation}
    C(A) = \sup_{t\in [0,T]}\|\exp(At)\|\,,
\end{equation}
where $T$ is the time to which we need to simulate the dynamics. Here, we prove some bounds on this norm when $A$ is a block $2\times 2$ matrix as is the case in all our applications here.

We first consider the case when $F_0=0$.
\begin{lemma}\label{lem:C(A)_F_0}
Suppose 
\begin{equation}
    A = \begin{pmatrix}
        F_1 & 0\\F_2 & F_3
    \end{pmatrix}\,.
\end{equation}
Then we have
\begin{equation}
    C(A)\leq C_d(1+C_d\|F_2\|T)\,,
\end{equation}
where 
\begin{equation}\label{eq:C_d}
    C_d=\max\{\sup_{t\in [0,T]}\|\exp(F_1t)\|,\sup_{t\in [0,T]}\|\exp(F_3t)\|\}\,.
\end{equation}
\end{lemma}
\begin{proof}
Split the matrix $A$ as follows
\begin{equation}\label{eq:A_split}
    A=\begin{pmatrix} F_1 & 0\\0 & F_3\end{pmatrix} + \begin{pmatrix} 0 & 0\\F_2 & 0\end{pmatrix}=A_1+A_2\,,
\end{equation}
where $A_1$ contains the dissipative part of the evolution and $A_2$ contains the nonlinear term. From the definition in \eq{C_d} above, we see that $C(A_1)= C_d$. Now using Dyson's formula (\cite{bhatia1996matrix} or \cite{kato2013perturbation}), we have
\begin{equation}
    \exp(At) = \exp(A_1t) + \int_0^t \exp(A_1(t-s))A_2\exp(As) ds\,.
\end{equation}
We can define a recursive version of this by using the formula to define $\exp(As)$. We get (see \cite{kato2013perturbation})
\begin{equation}
    \exp(At) = \sum_{n=0}^\infty U_n(t)\,,
\end{equation}
where
\begin{equation}
    U_{n+1}(t) = \int_0^t \exp(A_1(t-s))A_2 U_n(s) ds\,,
\end{equation}
and $U_0(t) = \exp(A_1t)$. It turns out that in our case, $U_n=0$ for $n\geq 2$. To see this, note that
\begin{equation}
    \exp(A_1t)A_2\exp(A_1t')A_2 = 0\,,
\end{equation}
since $A_1$ is a block diagonal matrix and $A_2$ is a strictly lower triangular block matrix. Therefore, we have
\begin{equation}\label{eq:exp(At)}
    \exp(At) = \exp(A_1t) + \int_0^t \exp(A_1(t-s))A_2\exp(A_1s) ds\,.
\end{equation}
Taking norms, we get 
\begin{equation}\label{eq:C(A)_F_0_eq}
    C(A) \leq C_d(1+C_d T\|F_2\|)\,,
\end{equation}
where we have used \eq{C_d} and the fact that $\|A_2\|\leq \|F_2\|$.
\end{proof}
Note that we get a similar result when $F_2=0$ instead of $F_0$. We now show a weaker result when both $F_2$ and $F_0$ are non-zero.
\begin{lemma}\label{lem:C(A)_F_neq_0}
Assume that $F_0\neq 0$ and $F_2\neq 0$. Also assume that
\begin{equation}
    \|\exp(At)\|\leq \exp(\mu t)\,.
\end{equation}
Then we have 
\begin{equation}
    \|\exp(At)\|\leq \exp((\mu + \|F_2\| + \|F_0\|)t)\,.
\end{equation}
\end{lemma}
\begin{proof}
As before, split $A$ as follows
\begin{equation}\label{eq:A_split2}
    A=\begin{pmatrix} F_1 & 0\\0 & F_3\end{pmatrix} + \begin{pmatrix} 0 & F_0\\F_2 & 0\end{pmatrix}=A_1+A_2\,,
\end{equation}
we have that $A_1$ contains the dissipative part of the evolution and $A_2$ contains the nonlinear and constant terms. From \cite{kato2013perturbation}, for any two matrices $B$ and $E$ if $\|\exp(Bs)\|\leq \exp(\beta s)$, then
\begin{equation}
    \|\exp((B+E)s)\|\leq \exp((\beta + \|E\|)s)\,.
\end{equation}
We now use the fact that $\|\exp(A_1t)\|\leq \exp(\mu t)$, where $\mu$ is the log norm of $A_1$. Using this, we see that
\begin{equation}
    \|\exp((A_1+A_2)t)\|\leq \exp((\mu + \|F_0\|+\|F_2\|)t)\,.
\end{equation}
\end{proof}

\begin{remark}
In order to have $C(A)$ be polynomially bounded i.e., $C(A)=\poly(n)$, we need
\begin{equation}\label{eq:C(A)_cond}
    \exp((\mu + \|F_0\| + \|F_2\|)t) = O(n^k)\,,
\end{equation}
for some $k$ for all $t\in [0,T]$. If $\mu+\|F_0\|+\|F_2\|>0$, then we need
\begin{equation}\label{eq:mu_cond}
    (\mu+\|F_0\|+\|F_2\|)T = O(k\log n)\,.
\end{equation}
If $\mu+\|F_0\|+\|F_2\|<0$, then the condition in \eq{C(A)_cond} is satisfied already.
\end{remark}

\begin{remark}\label{rem:R_comparison}
We compare the above quantity to $R$, which was defined in \cite{Liue2026805118} (see also \cite{krovi2022improved} for a generalized version) as follows.
\begin{equation}
    R=\frac{1}{|\mu|}\Big(\|F_2\|\|u_{in}\| + \frac{\|F_0\|}{\|u_{in}\|}\Big)\,,
\end{equation}
where $\mu$ is the log norm of $F_1$, which is assumed to be negative and $\|u_{in}\|$ is the norm of the initial state. This means we need to re-define $F_0$ and $F_2$ as $\bar{F_0}=F_0/\|u_{in}\|$ and $\bar{F_2}=\|u_{in}\|F_2$. The requirement that $R<1$ (by Carleman linearization) is equivalent to
\begin{equation}\label{eq:mu_old}
    \mu +\|\bar{F_2}\| + \|\bar{F_0}\| \leq 0\,.
\end{equation}
In contrast, here we need (from \eq{mu_cond})
\begin{equation}\label{eq:mu_new}
    \mu +\|F_2\| + \|F_0\| = \frac{1}{T}O(\log n)\,,
\end{equation}
which grows logarithmically in $n$, which is the number of bits needed to encode the problem i.e., $n=\log N$. The condition \eq{mu_new} is more general than (and includes) \eq{mu_old}. In many cases, it can lead to more nonlinearity and even positive log-norm $\mu$.
\end{remark}

\begin{remark}
When $\mu(A_1)<0$ and $F_0=0$, we have an even weaker constraint on the nonlinearity. From \eq{C(A)_F_0_eq}, we have
\begin{equation}
    1+\|F_2\|T \leq \poly(n)\,,
\end{equation}
since $C_d\leq 1$. This means that the strength of the nonlinear term can be as large as $\poly(n)/T$.
\end{remark}

\section{Quantum algorithms for classical Hamiltonians}\label{sec:class_ham_sim}
In this section, we present our quantum algorithms to solve the problem of preparing the final state after classical evolution and also to estimate the kinetic energy as described in \prb{class_mech}. In this approach, we essentially solve the differential equation coming from Hamilton's canonical equations in \eq{Ham_eq_diss}, namely an equation of the type
\begin{equation}
    \frac{dx}{dt} = A x + b\,,
\end{equation}
using the quantum ODE solver recently developed in \cite{krovi2022improved}. Specifically, we have the following result.
\begin{theorem}\label{thm:Ham_canon}
There exists a quantum algorithm that solves \prb{class_mech} of preparing a quantum state given oracle access to $M$, $V$, $R$ and an oracle to prepare a normalized version of $f$. Let $d$ be the dimension of these matrices and $s$ be the sparsity. Let
\begin{equation}
    k=O(\log T\|A\|(1+\frac{T\|f\|}{\|x(T)\|})\,,\hspace{0.1in}\text{and}\hspace{0.1in} \kappa_L = T\|A\|C(A)\,.
\end{equation}
The query complexity of this algorithm is given by 
\begin{equation}
    T_Q=O(g \kappa_L \log\kappa_L \poly(s,k,\log d,\log1/\epsilon))\,,
\end{equation}
and gate complexity $T_G$ is greater by a factor of at most
\begin{equation}
    O(\poly(k,\log(\frac{1}{\epsilon}),\log(\|A\|T))\,,
\end{equation}
where 
\begin{equation}
    C(A)\leq \max\{\|\sqrt{V}\|,\|\sqrt{M}\|\}\cdot\max\{\|\sqrt{V^{-1}}\|,\|\sqrt{M^{-1}}\|\}\,,
\end{equation}
\begin{equation}
    g=\frac{\max_{t\in [0,T]}\|x(t)\|}{\|x(T)\|}\,,
\end{equation}
where $x(t)$ is the solution at time $t$, $T$ is the final time and
\begin{equation}
    \|A\|\leq 2\max\{\|M^{-1}V\|,\|M^{-1}R\|,1\}\,.
\end{equation}
The query complexity of estimating the kinetic energy is
\begin{equation}
    T_K=O( \frac{1}{\epsilon}\kappa_L \poly(\log\kappa_L,s,k,\log d))\,,
\end{equation}
with the gate complexity greater by a fraction of at most $\poly(\log 1/\epsilon,\log T\|A\|,k)$.
\end{theorem}
\begin{proof}
Using the result from \cite{krovi2022improved}, we can obtain a quantum state $\ket{\psi}$ that is close to the solution of the equation
\begin{equation}\label{eq:Ham_eq_linear}
    \frac{dx}{dt} = A x + b\,,
\end{equation}
such that $\|\ket{\psi} - x(T)/\|x(T)\|\|\leq \epsilon$. To do so, we need a block encoding of $A$ and an oracle to prepare the normalized version of $b$ and the normalized initial state. A block-encoding of $A$ can implemented by first block-encoding the sparse matrices
\begin{equation}
    \begin{pmatrix}
        0&I\\-M^{-1}V&0
    \end{pmatrix}\,\hspace{0.1in}\text{and}\hspace{0.1in}
    \begin{pmatrix}
        0&0\\0&-M^{-1}R
    \end{pmatrix}\,,
\end{equation}
and then adding the two block-encoded unitaries. The initial state $x(0)$ in this basis can be prepared easily using oracles to prepare $q(0)$ and $\dot{q}(0)$ as normalized states similar to the description in \cite{krovi2022improved}. The oracle to prepare $b$ can be done using the method described in \cite{babbush2023exponential}. Recall that 
\begin{equation}
    b=\begin{pmatrix}
        0\\M^{-1}f
    \end{pmatrix}\,.
\end{equation}
The method in \cite{babbush2023exponential} which uses inequality testing from \cite{Ineq_testing} can construct an oracle for $b$ with error $\epsilon$ with $O(\sqrt{\|M^{-1}\|}\polylog(d,1/\epsilon))$ queries to the oracle for $f$.

Using this, the final state $\ket{\psi}$ can be obtained with query complexity 
\begin{equation}
    O(g T \|A\| C(A) \log(\|A\|C(A)T) \poly(s,k,\log d,\log(1/\epsilon)))\,,
\end{equation}
where
\begin{equation}
    C(A) = \sup_{t\in [0,T]}\exp(At)\,,
\end{equation}
and 
\begin{equation}
    g= \frac{\max_{t\in [0,T]}\|x(t)\|}{\|x(T)\|}\,.
\end{equation}
Next, we estimate $C(A)$ in terms of the norms of $M$, $R$, $V$ and $f$.
\paragraph{Dissipation $R=0$.} First, as a special case, we consider the case when $R=0$. In this case, for the classical mechanical system, we have from \eq{dx_dt} that 
\begin{equation}
    A=\begin{pmatrix}
        0&I\\
        -M^{-1}V& 0
    \end{pmatrix}\,.
\end{equation}
The complexity of the algorithm depends on $C(A)$, which we now estimate. We have
\begin{equation}
    A=I_M^{-1}JQI_M\,,
\end{equation}
where
\begin{align}
    &I_M = \begin{pmatrix}
        I & 0\\ 0&M
    \end{pmatrix}\,,
    &Q = \begin{pmatrix}
        V &0\\0&M^{-1}
    \end{pmatrix}\,.
\end{align}
Therefore, we have
\begin{equation}
    A=I_M^{-1}Q^{-1/2}(Q^{1/2}JQ^{1/2})Q^{1/2}I_M\,.
\end{equation}
This means
\begin{equation}
    \exp(At) = (\sqrt{Q}I_M)^{-1}\exp(Q^{1/2}JQ^{1/2}t)(\sqrt{Q}I_M)\,.
\end{equation}
Taking norms and using the fact that $\sqrt{Q}J\sqrt{Q}$ is an anti-Hermitian operator (as discussed in \sec{quad}), we get
\begin{equation}
    \|\exp(At)\| \leq \|(\sqrt{Q}I_M)^{-1}\|\|\exp(Q^{1/2}JQ^{1/2}t)\|\|\sqrt{Q}I_M\| = \kappa(\sqrt{Q}I_M)\,,
\end{equation}
where in the above equation, we used the fact that the exponential of $\sqrt{Q}J\sqrt{Q}$ has unit norm and where $\kappa(\sqrt{Q}I_M)$ is the condition number of $\sqrt{Q}I_M$. Finally, we have
\begin{equation}
    \sqrt{Q}I_M = \begin{pmatrix}
        \sqrt{V} &0\\0&\sqrt{M}
    \end{pmatrix}\,.
\end{equation}

\paragraph{Dissipation $R\neq 0$.} In this case, we have that 
\begin{equation}
    A=\begin{pmatrix}
        0 & I\\-M^{-1}V & -M^{-1}R
    \end{pmatrix}\,.
\end{equation}
To bound $C(A)$ in this case, first we split $A$ as follows
\begin{equation}
    A=\begin{pmatrix}
        0 & I\\-M^{-1}V & -M^{-1}R
    \end{pmatrix}
    = I_M^{-1}\Bigg[\begin{pmatrix}
        0&M^{-1}\\-V&0
    \end{pmatrix} + \begin{pmatrix}
        0 & 0\\0& -RM^{-1}
    \end{pmatrix}\Bigg]I_M \,.
\end{equation}
The first term inside the square brackets is just $JQ$ and let us denote the second term $A_2$. Suppose we conjugate $A_2$ by $Q^{1/2}$, we get
\begin{equation}
    Q^{1/2}A_2Q^{-1/2} = \begin{pmatrix}
        0 & 0 \\ 0 & -M^{-1/2}RM^{-1/2}
    \end{pmatrix}\,.
\end{equation}
Let
\begin{equation}
    A_R = \begin{pmatrix}
        0 & 0 \\ 0 & -M^{-1/2}RM^{-1/2}
    \end{pmatrix}\,,
\end{equation}
and similarly conjugating $JQ$, we get
\begin{equation}
    A_I = Q^{1/2}JQ^{1/2}
\end{equation}

Using this, we can see that
\begin{equation}
    A = I_M^{-1}Q^{-1/2}(Q^{1/2}JQ^{1/2})Q^{1/2}I_M + I_M^{-1}Q^{-1/2}A_RQ^{1/2}I_M \,.
\end{equation}
Therefore
\begin{equation}
    \exp(At) = (\sqrt{Q}I_M)^{-1} \exp(A_I + A_R) (\sqrt{Q}I_M)\,.
\end{equation}
Hence
\begin{equation}
    \|\exp(At)\| \leq \kappa(\sqrt{Q}I_M) \|\exp(A_R+A_I)\|\,.
\end{equation}
Now a result from \cite{bhatia1996matrix} states that for any matrix $B$
\begin{equation}
    \|\exp(B)\|\leq \|\exp((B+B^\dag)/2)\|\,.
\end{equation}
Letting $B=A_R+A_I$, we get $(B+B^\dag)/2 = A_R$ (because $A_I$ is anti-Hermitian). Therefore, we have
\begin{equation}
    C(A)\leq \kappa(I_M\sqrt{Q}) C(A_R)\,.
\end{equation}
Since we assume that $R$ is a positive matrix, we have $M^{-1}RM^{-1}$ is positive and hence $A_R$ is negative. This means $\|\exp(A_Rt)\|\leq 1$. Therefore, 
\begin{equation}\label{eq:C_A}
    C(A)\leq \kappa(\sqrt{Q}I_M)\,.
\end{equation}
From the block structure of $A$, we can bound $\|A\|$ as follows. We have
\begin{equation}
    \|A\| = \sup_{v:\|v\|=1} \|Av\|\,.
\end{equation}
We can assume block structure for $v$ as $v=(v_1,v_2)^T$ such that $\|v_1\|^2 + \|v_2\|^2=1$. We now have
\begin{align}
    \|A\| &= \sup_{(v_1,v_2)^T:\|v_1\|^2 + \|v_2\|^2=1} \|(v_2,-M^{-1}Vv_1-M^{-1}Rv_2)^T\| \\
    &\leq \sup_{(v_1,v_2)^T:\|v_1\|^2 + \|v_2\|^2=1} \sqrt{\|v_2\|^2 + \max\{\|M^{-1}V\|,\|M^{-1}R\|\}^2\|v_1+v_2\|^2}\\
    &\leq \sup_{(v_1,v_2)^T:\|v_1\|^2 + \|v_2\|^2=1} \sqrt{2}\max\{\|M^{-1}V\|,\|M^{-1}R\|,1\}\|v_1+v_2\|\,.
\end{align}
We also have $\|v_1+v_2\|^2\leq 2(\|v_1\|^2+\|v_2\|^2)$. Using this, we can see that
\begin{equation}
    \|A\|\leq 2\max\{\|M^{-1}V\|,\|M^{-1}R\|,1\}\,.
\end{equation}
\paragraph{Estimation of kinetic energy.} The kinetic energy of the system at a given time is defined as
\begin{equation}
    K(t)=\frac{1}{2}\dot{q}(t)^\dag M \dot{q}(t)\,.
\end{equation}
To estimate the kinetic energy of the final state, we need to use the so called history state which is the quantum state produced by the algorithm (from which we can get the normalized solution state by measuring the time register). Here we will not measure the time register, but perform a Hadamard test on the history state. This state can be written as 
\begin{equation}
    \ket{\psi} = \frac{1}{N_\psi}\Bigg(\sum_{i=0}^m \ket{i,x_i}\Bigg)\,,
\end{equation}
where $N_\psi$ is the normalization and $m=T\|A\|$ denotes the number of time steps needed to obtain an accurate solution (to within an additive error). Using \thm{norm_estimate}, we can estimate the norm of the above state to some multiplicative error $\gamma$ i.e., there exists a quantum algorithm that produces $\hat{N}_\psi$ such that
\begin{equation}\label{eq:N_hat_psi}
    |\hat{N}_\psi - N_\psi|\leq \gamma N_\psi\,.
\end{equation}
We can now measure the state $\ket{\psi}$ using the following procedure. We first block-encode as $U_M$, the matrix 
\begin{equation}
    \frac{1}{2}\ket{m}\bra{m}\otimes \begin{pmatrix}
        0 & 0\\ 0 & M
    \end{pmatrix}
\end{equation}
Then we need to estimate the inner product
\begin{equation}
    \frac{1}{2}\dot{q}^\dag M \dot{q}=\frac{1}{2}\begin{pmatrix}
        q&\dot{q}
    \end{pmatrix}
    \begin{pmatrix}
        0 & 0\\ 0 & M
    \end{pmatrix}
    \begin{pmatrix}
        q\\ \dot{q}
    \end{pmatrix}\,,
\end{equation}
by estimating the inner product
\begin{equation}
    \bra{\psi}U_M\ket{\psi} = \frac{1}{N^2_\psi\|M\|}\frac{1}{2}\hat{\dot{q}}^\dag M \hat{\dot{q}}\,,
\end{equation}
where $\hat{\dot{q}}$ is the velocity obtained from the numerical solution and $\|M\|$ is the sub-normalization in the block-encoding. In other words, the inner product is a scaled version of the kinetic energy. We can now use the estimate of the norm to get the kinetic energy $\hat{K} = \|M\|\hat{N}_\psi^2\bra{\psi}U_M\ket{\psi}$. The total error in this estimate is
\begin{equation}
    |\hat{K} - K| = \frac{1}{2}\Bigg|\frac{\hat{N}^2_\psi}{N^2_\psi}\hat{\dot{q}}^\dag M \hat{\dot{q}}-\dot{q}^\dag M \dot{q}\Bigg|\,.
\end{equation}
This can be written as
\begin{equation}
    |\hat{K} - K| = \frac{1}{2}\Bigg|\frac{\hat{N}^2_\psi}{N^2_\psi}(\hat{\dot{q}}^\dag M \hat{\dot{q}}-\dot{q}^\dag M \dot{q}) + (\frac{\hat{N}^2_\psi}{N^2_\psi}-1)\dot{q}^\dag M \dot{q}\Bigg|\,.
\end{equation}
Using \eq{N_hat_psi}, we have $\hat{N}_\psi\leq (1+\gamma)N_\psi$ and this means
\begin{equation}
    |\hat{N}_\psi^2-N_\psi^2|=|\hat{N}_\psi-N_\psi|(\hat{N}_\psi+N_\psi)\leq \gamma(2+\gamma)N_\psi^2\,.
\end{equation}
Therefore,
\begin{equation}\label{eq:kinetic_diff}
    |\hat{K} - K| \leq \frac{1}{2}(1+\gamma(2+\gamma))|(\hat{\dot{q}}^\dag M \hat{\dot{q}}-\dot{q}^\dag M \dot{q})| + \frac{1}{2}\gamma(2+\gamma)|\dot{q}^\dag M \dot{q}|\,.
\end{equation}
Now let us focus on the difference of inner products in the first term. This can be written as
\begin{equation}
    |\hat{\dot{q}}^\dag M \hat{\dot{q}} - \dot{q}^\dag M \dot{q}| =\sum_i m_i |\hat{\dot{q}}_i^2 - \dot{q}_i^2|\leq m_{\max}  \sum_i|\hat{\dot{q}}_i^2 - \dot{q}_i^2|\,,
\end{equation}
where $m_{\max}$ is the maximum value of the diagonal matrix of masses $M$. Using Cauchy-Schwartz inequality, we can write this as 
\begin{equation}
    m_{\max}  \sum_i|\hat{\dot{q}}_i^2 - \dot{q}_i^2|\leq m_{\max} \sqrt{\sum_i(\hat{\dot{q}}_i - \dot{q}_i)^2}\sqrt{\sum_i(\hat{\dot{q}}_i + \dot{q}_i)^2} = m_{\max} \|\hat{\dot{q}} - \dot{q}\|\|\hat{\dot{q}} + \dot{q}\|\,,
\end{equation}
where $\|\cdot\|$ is the $\ell_2$ norm. Using the fact that
\begin{equation}
    \|\hat{\dot{q}} - \dot{q}\|\leq \epsilon_1\|\dot{q}\|\,,
\end{equation}
for some $\epsilon_1$ we get 
\begin{equation}
    |\hat{\dot{q}}^\dag M \hat{\dot{q}} - \dot{q}^\dag M \dot{q}|\leq m_{\max}\epsilon_1(2+\epsilon_1)\|\dot{q}\|^2\,.
\end{equation}
Putting it together and choosing $\gamma$ such that $\gamma= \epsilon_1 = \epsilon/8$, we get
\begin{equation}
    |\hat{K} - K|\leq \epsilon K_{\max}\,,
\end{equation}
where
\begin{equation}
    K_{\max} = m_{\max}\|\dot{q}\|_2^2\,.
\end{equation}
The gate complexity of kinetic energy estimation follows from the complexity of preparation of the state $\ket{\psi}$, norm estimation of $N_\psi$, block encoding $U_M$ and the Hadamard test.
\end{proof}

\begin{remark}
The run-time for the case with zero damping i.e., when $R=0$, depends on the condition number of the matrix $\sqrt{Q}$. The fact that for the encoding used here, the run-time scales as the inverse of the smallest eigenvalue is already pointed out in \cite{babbush2023exponential}, where it was argued that it is at least linear. Here we show that it is no more than linear (up to log factors).

It is interesting to note that the problem shown to be $\BQP$ complete in \cite{babbush2023exponential} also has a constant condition number for $\sqrt{Q}$ and hence the above algorithm and encoding would be efficient for this case.
\end{remark}
\begin{remark}
It is also worth noting that the run-time of the quantum algorithm from \cite{krovi2022improved} for the linear differential equation \eq{Ham_eq_linear} is dependent on the norm of the exponential of $A$ (denoted $C(A)$ here) and does not require $A$ to have a negative log-norm.
\end{remark}
Next we give a result that improves the above kinetic energy estimation by giving a multiplicative or relative approximation to the kinetic energy if there is no source term i.e., $b=0$.
\begin{corollary}
If $b=0$ in \eq{Ham_eq_linear}, then there exists a quantum algorithm that outputs an estimate $\hat{K}$ for the kinetic energy such that
\begin{equation}
    |K-\hat{K}|\leq \epsilon K\,,
\end{equation}
where $K$ is the actual kinetic energy.
\end{corollary}
\begin{proof}
From \eq{kinetic_diff} in the proof of the $b\neq 0$ case, we have that the estimate $\hat{K}$ satisfies
\begin{equation}
    |K-\hat{K}|\leq \frac{1}{2}(1+\gamma(2+\gamma))\sum_im_i|\dot{q}_i^2 - \hat{\dot{q}}_i^2| + \frac{1}{2}\gamma(2+\gamma)|\dot{q}^\dag M \dot{q}|\,.
\end{equation}
We show in \lem{q_i_diff} in \appx{auxil_lemma} that for any $i$ 
\begin{equation}
    |\dot{q}_i - \hat{\dot{q}}_i|\leq \epsilon_1\dot{q}_i\,,
\end{equation}
for some constant $\epsilon_1$. Using this, we get
\begin{equation}
    |\dot{q}_i^2 - \hat{\dot{q}}_i^2|=|\dot{q}_i - \hat{\dot{q}}_i||\dot{q}_i + \hat{\dot{q}}_i|\leq \epsilon_1(2+\epsilon_1)\dot{q}_i^2\,,
\end{equation}
which gives
\begin{equation}
    |K-\hat{K}|\leq (\epsilon_1(2+\epsilon_1)(1+\gamma(2+\gamma))+\gamma(2+\gamma))K\,.
\end{equation}
By choosing $\gamma=\epsilon_1=\epsilon/16$, we get
\begin{equation}
    |K-\hat{K}|\leq \epsilon K\,.
\end{equation}
\end{proof}

\subsection{Improved kinetic energy estimation}
In this subsection, we explain how to use Quantum Phase Estimation (QPE) and Qubitization \cite{Low2019hamiltonian} and the improvements to this implemented in \cite{npj2018, PRXQuantum.1.020312}. In this approach, we have an algorithm whose complexity does not have a dependence on the condition number of $\sqrt{Q}$ as with the previous approach. The downside is that it has a worse dependence on the error. Recall that in the $y$ basis, we have the differential equation given by \eq{y_basis} reproduced below.
\begin{equation}
    \frac{dy}{dt} = \begin{pmatrix}
        0 & \sqrt{VM^{-1}}\\
        -\sqrt{M^{-1}V} & -\sqrt{M^{-1}}R\sqrt{M^{-1}}
    \end{pmatrix}y+b = \tilde{A}y+b\,.
\end{equation}
In order to solve this equation, we need to be able to block-encode $\sqrt{M^{-1}V}$. This can be done using the following result shown in \cite{babbush2023exponential}.
\begin{theorem}\label{thm:sqrt_A}
Given oracle access to the sparse matrix $A$, one can implement a block-encoding of $\sqrt{A}$ with query complexity
\begin{equation}
   T_Q(A)= O(\|A\|_{\max}s\log(1/\epsilon)\frac{1}{\epsilon}\min(\sqrt{A^{-1}},\frac{1}{\epsilon}))\,,
\end{equation}
and gate complexity greater by a factor of at most 
\begin{equation}
    O(\polylog(\frac{d}{\epsilon}))\,.
\end{equation}
\end{theorem}
Using this result, we prove the following theorem.
\begin{theorem}\label{thm:class_mech_qpe}
There exists a quantum algorithm that solves \prb{class_mech} of estimating to within an additive error of $\epsilon$, the kinetic energy of a classical mechanical system described by its equations of motion as follows
\begin{equation}
    M\ddot{q} + R\dot{q} + Vq +f = 0\,,
\end{equation}
We assume that we are given oracle access to $M$, $V$, $R$ and an oracle to prepare a normalized version of $f$. Let $d$ be the dimension of these matrices and $s$ be the sparsity. Let
\begin{equation}
    k=O(\log T\|\Tilde{A}\|(1+\frac{T\|f\|}{\|x(T)\|})\,,\hspace{0.1in} m=T\|\Tilde{A}\|\,.
\end{equation}
The query complexity of this algorithm is given by 
\begin{equation}
    T_Q=O(\frac{m}{\epsilon^2} \log(m) \poly(s,k,\log d))\,,
\end{equation}
and gate complexity $T_G$ is greater by a factor of at most
\begin{equation}
    O(\poly(k,\log(\frac{1}{\epsilon}),\log(m))\,,
\end{equation}
where
\begin{equation}
    \|\Tilde{A}\|\leq 2\max\{\|M^{-1}V\|,\|M^{-1}R\|,1\}\,.
\end{equation}
\end{theorem}
\begin{proof}
We can first block encode $\sqrt{M^{-1}V}$ as a unitary $U_1$ using \thm{sqrt_A}. Using this block-encoding, we can block encode the following matrix
\begin{equation}
    \begin{pmatrix}
        0 & \sqrt{VM^{-1}}\\
        -\sqrt{M^{-1}V} & 0
    \end{pmatrix}\,.
\end{equation}
This can be implemented by the following circuit
\begin{equation}
    U_2 = cU_1^\dag(iY\otimes I)cU_1\,,
\end{equation}
where $Y$ is the Pauli matrix on a single qubit (also denoted $\sigma_y$) and $cU_1$ is the controlled version of the unitary $U_1$ controlled by the first qubit. The identity in the tensor product is of the same dimension as $U_1$. Next, we block encode the sparse matrix 
\begin{equation}
    \begin{pmatrix}
        0&0\\0&-\sqrt{M^{-1}}R\sqrt{M^{-1}}
    \end{pmatrix}\,,
\end{equation}
as $U_3$. Finally, we use block-encoded matrix arithmetic to perform $U_2+U_3$. To extract the kinetic energy, we first write it as follows. Initial state preparation and the oracle to prepare $b$ can be done as before using the techniques from \cite{babbush2023exponential}. Recall that $y$ is
\begin{equation}
    y=\begin{pmatrix}
        \sqrt{V}q\\
        \sqrt{M^{-1}}p
    \end{pmatrix}\,.
\end{equation}
Therefore, the kinetic energy can be written as
\begin{equation}
    K=\frac{1}{2}\dot{q}^TM \dot{q} = \frac{1}{2}y^TPy\,,
\end{equation}
where 
\begin{equation}
    P=\begin{pmatrix}
    0 & 0 \\0&I    
    \end{pmatrix}\,,
\end{equation}
which is the projector onto the second block of $y$. Similar to the previous approach, we can implement a block-encoding (denoted $U_P$) of $P$ tensored with the time register as follows (one can also implement a unitary that puts a sign on the basis states corresponding to $P$).
\begin{equation}
    \ket{m}\bra{m}\otimes P\,.
\end{equation}
In a manner similar to the proof of \thm{Ham_canon}, we prepare the history state that has the unnormalized approximate solution $\hat{y}$ at time $T$. This state can be written as 
\begin{equation}
    \ket{\phi} = \frac{1}{N_\phi}\sum_{i=0}^m \ket{i,y_i}\,,
\end{equation}
and the query complexity of preparing this state is
\begin{equation}
    O(T \|\Tilde{A}\| C(\Tilde{A}) \log(\|\Tilde{A}\|C(\Tilde{A})T) \poly(s,k,\log d))\,,
\end{equation}
where
\begin{equation}
    C(\Tilde{A}) = \sup_{t\in [0,T]}\exp(\Tilde{A}t)\,.
\end{equation}
We now estimate the quantity $C(\Tilde{A})$ in this case. First, we can write $\Tilde{A}$ as
\begin{equation}\label{eq:A_1A_2}
    \Tilde{A}=\begin{pmatrix}
        0 & \sqrt{VM^{-1}}\\
        -\sqrt{M^{-1}V} & -\sqrt{M^{-1}}R\sqrt{M^{-1}}
    \end{pmatrix} = \begin{pmatrix}
        0 & \sqrt{VM^{-1}}\\
        -\sqrt{M^{-1}V} & 0
    \end{pmatrix} + \begin{pmatrix}
        0 & 0\\
        0 & -\sqrt{M^{-1}}R\sqrt{M^{-1}}
    \end{pmatrix}\,.
\end{equation}
We denote the terms in the above sum $\Tilde{A}_I$ and $\Tilde{A}_R$. Now using the fact that
\begin{equation}
    \|\exp(\Tilde{A}t)\|\leq \|\exp(\frac{\Tilde{A}+\Tilde{A}^\dag}{2}t)\|\,,
\end{equation}
and the fact that the first term in the sum on the right hand side of \eq{A_1A_2} i.e., $\Tilde{A}_I$ is anti-Hermitian, we get
\begin{equation}
    \|\exp(\Tilde{A}t)\|\leq \|\exp(\Tilde{A}_Rt)\|\,.
\end{equation}
Now since $\Tilde{A}_R$ is negative definite, we have $\|\exp(\Tilde{A}t)\|\leq 1$, which makes $C(\Tilde{A})\leq 1$.

Next, we estimate the norm $N_\phi$ using \thm{norm_estimate}. Finally, we perform the Hadamard test to estimate the inner product of $\ket{\phi}$ with $U_P$. This produces the following estimate of the normalized kinetic energy.
\begin{equation}
    \hat{K} = \frac{N_\phi}{2}\bra{\phi}U_P\ket{\phi}\,.
\end{equation}
The error analysis is similar to the one in the proof of \thm{Ham_canon}. The main difference is that the block-encoding of $\sqrt{M^{-1}V}$ leads to a query complexity that depends on $1/\epsilon_{PE}$, where $\epsilon_{PE}$ is the phase estimation error. Choosing $\epsilon_{PE}=\epsilon$, we get the quadratic dependence of the run-time on the error. The overall query complexity is as stated in the theorem.
\end{proof}

\subsection{Coupled oscillators with damping}
In this section, we focus specifically on coupled oscillators with damping and forcing. Using the specific structure of the matrices in this problem, in \cite{babbush2023exponential}, a different square-root of $Q$ for this system was shown to be efficiently implementable as a block-encoding. There, it was shown that there exists a matrix $B$ which is $d\times d^2$ such that $B^\dag B=\sqrt{M^{-1}}V\sqrt{M^{-1}}$ and that this matrix can be block-encoded without the overhead of $1/\epsilon$ that phase-estimation requires. Here we show that one can include damping and solve this problem using differential equation solvers. Going back to the $x$ basis, we have the equation
\begin{equation}
    \frac{dx}{dt}= Ax+b\,,
\end{equation}
where 
\begin{equation}
    A=\begin{pmatrix}
        0&I\\-M^{-1}V&-M^{-1}R
    \end{pmatrix}\,.
\end{equation}
We now define the basis, which we denote $\tilde{x}$, as follows.
\begin{equation}
    \tilde{x} = \begin{pmatrix}
        \sqrt{M} &0\\0&\sqrt{M}
    \end{pmatrix}x\,.
\end{equation}
In this basis, the differential equation becomes
\begin{equation}
    \frac{d\tilde{x}}{dt} = 
    \begin{pmatrix}
        \sqrt{M} &0\\0&\sqrt{M}
    \end{pmatrix} 
    \begin{pmatrix}
        0&I\\-M^{-1}V&-M^{-1}R
    \end{pmatrix}
    \begin{pmatrix}
        \sqrt{M^{-1}} &0\\0&\sqrt{M^{-1}}
    \end{pmatrix} \tilde{x} + \begin{pmatrix}
        \sqrt{M} &0\\0&\sqrt{M}
    \end{pmatrix} b
    \,.
\end{equation}
Since $b=(0,\sqrt{M}f)^T$, we have
\begin{equation}
    \frac{d\tilde{x}}{dt} =
    \begin{pmatrix}
        0&I\\-M^{-1}VM^{-1}&-M^{-1}RM^{-1}
    \end{pmatrix}\tilde{x} - \begin{pmatrix}
        0\\ \sqrt{M^{-1}}f
    \end{pmatrix}\,.
\end{equation}
We will write this as 
\begin{equation}
    \frac{d\tilde{x}}{dt} = (J\tilde{Q} + \tilde{A}_R)\tilde{x} +\tilde{b}\,,
\end{equation}
where $J$ was defined before and
\begin{equation}
    \tilde{Q} = 
    \begin{pmatrix}
        M^{-1}VM^{-1}&0\\0&I
    \end{pmatrix}\,,
\end{equation}
and 
\begin{equation}
    \tilde{A}_R = \begin{pmatrix}
        0&0\\0&-M^{-1}RM^{-1}
    \end{pmatrix}\,.
\end{equation}
Now define the matrix $B$ as follows (see \cite{babbush2023exponential}).
\begin{equation}
    BB^\dag = \sqrt{M^{-1}}V\sqrt{M^{-1}}\,,
\end{equation}
where $B$ is a $d\times d(d+1)/2$ dimensional matrix (recall that $V$ is a $d\times d$ matrix). Let us denote $d_1=d(d+1)/2$ for brevity. Now define the $2d\times (d_1+d)$ matrix $\tilde{B}$ as follows.
\begin{equation}
    \tilde{B} =
    \begin{pmatrix}
        B_{d\times d_1}&0_{d\times d}\\0_{d\times d_1}&I_{d\times d}
    \end{pmatrix}\,.
\end{equation}
Note that we have
\begin{equation}
    \tilde{B}\tilde{B}^\dag = 
    \begin{pmatrix}
        BB^\dag & 0_{d\times d}\\
        0_{d\times d}&I_{d\times d}
    \end{pmatrix}\,,
\end{equation}
which is exactly $\tilde{Q}$. Now define the basis $\tilde{y}$ as $\tilde{B}^\dag\tilde{x}$. In this basis, we have
\begin{equation}
    \frac{d\tilde{y}}{dt} = \tilde{B}^\dag(J\tilde{B}\tilde{B}^\dag+\tilde{A}_R)\tilde{x}+\tilde{B}^\dag\tilde{b}\,.
\end{equation}
It is easy to check that 
\begin{equation}
    \tilde{B}^\dag\tilde{A}_R=\tilde{A}_R=\tilde{A}_R\tilde{Q} = \tilde{A}_R\tilde{B}^\dag\tilde{B}\,.
\end{equation}
Using this, we get
\begin{equation}
    \frac{d\tilde{y}}{dt} = (\tilde{B}^\dag J\tilde{B}+\tilde{A}_R)\tilde{y}+\tilde{B}^\dag\tilde{b}\,,
\end{equation}
where 
\begin{equation}
    \tilde{B}^\dag\tilde{b}=
    \begin{pmatrix}
    0\\ B^\dag\sqrt{M^{-1}}f    
    \end{pmatrix}\,.
\end{equation}
Now that we have the differential equation in the $\tilde{y}$ basis, we can prove the following result. In the following, rather than make queries to a oracle for $V$, we assume access to an oracle for spring constants $S$ as in \cite{babbush2023exponential} (where it is denoted $K$) as well as an oracle for $M$.
\begin{theorem}
When we restrict \prb{class_mech} to the specific case of damped coupled oscillators with forcing, there exists a quantum algorithm that produces an estimate of the kinetic energy at time $T$ to within an additive error $\epsilon$ that has query complexity
\begin{equation}
    T_Q=O(\frac{m}{\epsilon}\poly(k,\log \frac{1}{\epsilon},\log m))
\end{equation}
and gate complexity greater by a factor of at most
\begin{equation}
    \polylog(m,\frac{1}{\epsilon},d,\frac{M_{max}}{M_{\min}})
\end{equation}
where
\begin{equation}
    m=Ts\max\bigg(\frac{R_{\max}}{M_{\min}},\sqrt{\frac{S_{\max}}{M_{\min}}}\bigg)\,,
\end{equation}
and
\begin{equation}
    k=O(\log m(1+\frac{T\|f\|}{\|\tilde{y}(T)\|})\,.
\end{equation}
\end{theorem}
\begin{proof}
The main idea is to solve the differential equation for $\tilde{y}$ and obtain the history state. From the history state, we can get the kinetic energy by block-encoding the projector onto the velocity space tensored with the final time i.e.,
\begin{equation}
    U_P = \bra{m}\ket{m}\otimes P_V\,,
\end{equation}
where $P_V$ is the projector on to the velocity component of $\tilde{y}$. The differential equation can be solved using the algorithm of \cite{krovi2022improved} provided some conditions are satisfied. First, we need a block-encoding of the matrix
\begin{equation}\label{eq:A_y_tilde}
    \tilde{B}^\dag J\tilde{B}+\tilde{A}_R = 
    \begin{pmatrix}
        0&B^\dag\\
        -B&-\sqrt{M^{-1}}R\sqrt{M^{-1}}
    \end{pmatrix}
    \,.
\end{equation}
This can be done by block-encoding $\tilde{B}^\dag J\tilde{B}$ and $\tilde{A}_R$ separately and adding them using matrix arithmetic. The block-encoding of $\tilde{B}^\dag J\tilde{B}$ is described in \cite{babbush2023exponential} and takes $O(\log d + \log^2(M_{\max}/M_{\min}\epsilon))$ and provides a block-encoding with parameters
\begin{equation}
    (\sqrt{2sS_{\max}/M_{\min}},2\log d+\log1/\epsilon,\epsilon)\,.
\end{equation}
The initial state preparation and an oracle for $\sqrt{M^{-1}}f$ is also similar to \cite{babbush2023exponential}. Finally, we need to compute the norm of the matrix exponential of the matrix in \eq{A_y_tilde}. We can see that this matrix has an anti-Hermitian and a Hermitian components. The anti-Hermitian component is $\tilde{B}^\dag J\tilde{B}$ and the Hermitian part is $\tilde{A}_R$. As in the previous sections, we can bound the norm of the exponential by
\begin{equation}
    \exp((\tilde{B}^\dag J\tilde{B}+\tilde{A}_R)t)\leq \exp(\tilde{A}_Rt)\leq 1\,,
\end{equation}
since $R$ is a positive matrix. This gives the stated complexity. The analysis of the error obtained in the kinetic energy estimation is the same as before.
\end{proof}

\subsection{Hardness of kinetic energy estimation}\label{sec:hardness}
In this subsection, we prove the main hardness result that estimating the kinetic energy of a $d$ dimensional coupled oscillator system governed by \eq{M_R_V_b_eq} i.e.,
\begin{equation}
    M\ddot{q} + R\dot{q} + Vq = 0\,,
\end{equation}
is $\BQP$ hard when the damping strength (or the norm of $R$) is bounded by an inverse polynomial in $\log d$. To prove this result, we need the following result proved in \cite{babbush2023exponential}, which shows the following. Here we give an equivalent formulation of the hardness result.
\begin{theorem}\label{thm:R_eq_0_hardness}
There exists a family of instances of coupled oscillators such that it is $\BQP$ hard to estimate the normalized kinetic energy $K/E$ to within additive precision $\epsilon$, where $K$ is the kinetic energy and $E$ is the total energy.
\end{theorem}
\begin{remark}
The above formulation follows from the fact that in \cite{babbush2023exponential}, it was shown that the normalized kinetic energy is equal to the output probability of a quantum circuit and estimating the latter to within additive precision is $\BQP$ hard.   
\end{remark}
\begin{remark}
It was shown in \cite{babbush2023exponential} that the instances can be chosen such that the masses are all equal to identity and the matrix of coupling constants $V$ has a condition number independent of the number of oscillators. The initial conditions can also be chosen to be $q(0)=(0,\dots,0)^T$ and $\dot{q}(0)=(1,-1,0,\dots,0)^T$.
\end{remark}
When the masses are all unity, the kinetic energy can be written as
\begin{equation}
    K=\frac{1}{2}\dot{q}^\dag \dot{q}\,,
\end{equation}
where $\dot{q}$ is given by
\begin{equation}
    \dot{q}(t) = \cos(\sqrt{V}t)\dot{q}(0)\,.
\end{equation}
Before we prove our result, we need the following lemma which derives a formula for the exponential of a matrix of the form $JQ$, where $Q$ is block diagonal.
\begin{lemma}\label{lem:exp_JQ}
Let 
\begin{equation}
    Q=\begin{pmatrix}
        V&0\\0&I
    \end{pmatrix}\,,
\end{equation}
and (as before)
\begin{equation}
    J=\begin{pmatrix}
        0&I\\-I&0
    \end{pmatrix}\,.
\end{equation}
Then we have
\begin{equation}
    \exp(JQt) = 
    \begin{pmatrix}
    \cos(\sqrt{V}t) & \sqrt{V^{-1}}\sin(\sqrt{V}t)\\
    -\sqrt{V}\sin(\sqrt{V}t) & \cos(\sqrt{V}t)
    \end{pmatrix}\,.
\end{equation}
\end{lemma}
\begin{proof}
Exponentiating $JQt$, we have
\begin{equation}\label{eq:exp}
    \exp(JQt) = \sum_{k=0}^\infty \frac{(JQt)^k}{k!}\,.
\end{equation}
To calculate this, observe that
\begin{equation}
    (JQ)^2 = 
    \begin{pmatrix}
    0 & I\\-V & 0    
    \end{pmatrix}
    \begin{pmatrix}
    0 & I\\-V & 0    
    \end{pmatrix} = 
    \begin{pmatrix}
    -V & 0\\0 & -V    
    \end{pmatrix} =
    I\otimes (-V)\,,
\end{equation}
where $I$ is a $2\times 2$ identity matrix. Now consider one of the terms of the right hand side of $\eq{exp}$. When $k=2\ell$ is even, we have
\begin{equation}
    (JQt)^k = t^{2\ell} (-1)^\ell (I\otimes V^\ell) \,.
\end{equation}
Therefore, the even terms in \eq{exp} are
\begin{equation}
    \sum_{\ell=0}^\infty \frac{t^{2\ell}(-1)^\ell}{(2\ell)!} (I\otimes \sqrt{V}^{2\ell}) = I\otimes \cos(\sqrt{V}t)\,.
\end{equation}
Similarly, the odd terms can be calculated to be
\begin{equation}
    (JQt)^{2\ell+1} = JQ (JQ)^{2\ell} = JQ t^{2\ell+1} (-1)^\ell (I\otimes \sqrt{V}^{2\ell}) \,.
\end{equation}
This can be written as
\begin{equation}
    t^{2\ell+1} (-1)^\ell J 
    \begin{pmatrix}
    \sqrt{V} &0 \\0 & \sqrt{V^{-1}}    
    \end{pmatrix}
    (I\otimes \sqrt{V}^{2\ell+1})\,.
\end{equation}
Therefore, the odd terms can be written as
\begin{equation}
    J\begin{pmatrix}
    \sqrt{V} &0 \\0 & \sqrt{V^{-1}}    
    \end{pmatrix}
    \sum_{\ell=0}^\infty \frac{t^{2\ell+1}(-1)^\ell}{(2\ell+1)!}(I\otimes \sqrt{V}^{2\ell+1}) = J\begin{pmatrix}
    \sqrt{V} &0 \\0 & \sqrt{V^{-1}}    
    \end{pmatrix} (I\otimes \sin(\sqrt{V}t))\,.
\end{equation}
Putting them together, we get 
\begin{equation}
    \exp(JQt) = \begin{pmatrix}
        \cos(\sqrt{V}t) & \sqrt{V^{-1}}\sin(\sqrt{V}t)\\
        -\sqrt{V}\sin(\sqrt{V}t)&\cos(\sqrt{V}t)
    \end{pmatrix}\,.
\end{equation}
\end{proof}

From the above lemma, we see that for the initial state given by $q(0)=(0,\dots,0)^T$ and $\dot{q}(0)=(1,-1,0,\dots,0)^T$, we have
\begin{equation}
    \begin{pmatrix}
        q(t)\\\dot{q}(t)
    \end{pmatrix} = \exp(JQt)\begin{pmatrix}
        q(0)\\\dot{q}(0)
    \end{pmatrix}\,,
\end{equation}
where
\begin{equation}
    Q = \begin{pmatrix}
        V & 0 \\0 & I
    \end{pmatrix}\,.
\end{equation}
Using \lem{exp_JQ}, we can write
\begin{equation}
    \dot{q}(t)=\cos(\sqrt{V}t)\dot{q}(0)\,.
\end{equation}
We can now prove our main hardness result.
\begin{theorem}\label{thm:hardness}
For the classical mechanical system defined in \prb{class_mech}, when we restrict to damped coupled oscillators, let $K(T)$ denote the kinetic energy at time $T$ and $E$ denote the total energy at time $t=0$. There exists a family of instances such that it is $\mathsf{BQP}$ hard to estimate the normalized kinetic energy $K(T)/E$ to within an additive error $\epsilon$ when $T=\polylog(d)$ and $\|R\|T\leq \epsilon$ (where $d$ is the dimension of the system).
\end{theorem}
\begin{proof}
When all the mass terms are unity, the kinetic energy can be written as 
\begin{equation}
    K(t) = \frac{1}{2}\dot{q}^T(t)\dot{q}(t)\,,
\end{equation}
which can be written as
\begin{equation}
    K(t)=\frac{1}{2}x(t)^T Px(t)\,,
\end{equation}
where $x(t) = (q(t),\dot{q}(t))^T$ and $P$ is the projector on the subspace spanned by $\dot{q}(t)$. When $R=0$, we also have that
\begin{equation}
    x(t) = \exp(JQt)x(0)\,,
\end{equation}
which gives
\begin{equation}
    K(t)=\frac{1}{2}x(0)^T\exp(-QJt)P\exp(JQt)x(0)\,,
\end{equation}
using the fact that $J^T=-J$ and $Q^T=Q$. Similarly, for non-zero $R$, we have
\begin{equation}\label{eq:K_R}
    K_R(t) = \frac{1}{2}x(0)^T\exp(-QJt+A_Rt)P\exp(JQt + A_Rt)x(0)\,,
\end{equation}
where 
\begin{equation}
    A_R = \begin{pmatrix}
        0 &0\\0&-R
    \end{pmatrix}\,.
\end{equation}
Now using Dyson's formula, we have
\begin{equation}
    \exp(JQt + A_Rt) = \exp(JQt) + \int_0^t \exp(JQ(t-s))A_R\exp(JQs+A_Rs)ds\,,
\end{equation}
and similarly for $\exp(-QJt+A_Rt)$. Plugging this into \eq{K_R} above, we get 
\begin{align}
    &2K_R(t) \\
    &= x(0)^T\exp(-QJt)P\exp(JQt)x(0)\label{eq:first_term}\\
    & + \int_0^t x(0)^T\exp(-QJt)P\exp(JQ(t-s))A_R\exp(JQs+A_Rs)x(0)ds\label{eq:second_term}\\
    & + \int_0^t x(0)^T\exp(-QJ(t-s))A_R\exp(-QJs+A_Rs)P\exp(JQt)x(0)ds\label{eq:third_term}\\
    & + \int_0^t\int_0^t x(0)^T\exp(-QJ(t-s_1))A_R\exp(-QJs_1+A_Rs_1)P\nonumber \\
    &\exp(JQ(t-s_2)A_R\exp(JQs_2+A_Rs_2)x(0)ds_1ds_2\label{eq:fourth_term}\,.
\end{align}
The first term \eq{first_term} can be seen to be $K(t)$. The second term \eq{second_term} is
\begin{align}
    & \int_0^t x(0)^T\exp(-QJt)P\exp(JQ(t-s))A_R\exp(JQs+A_Rs)x(0)ds\\
    &= \int_0^t x(0)^T\sqrt{Q}\exp(-\sqrt{Q}J\sqrt{Q}t)\sqrt{Q^{-1}}P\sqrt{Q^{-1}}\exp(\sqrt{Q}J\sqrt{Q}(t-s))\nonumber \\
    &\sqrt{Q}A_R\sqrt{Q^{-1}}\exp(\sqrt{Q}J\sqrt{Q}s+\sqrt{Q}A_R\sqrt{Q^{-1}}s)\sqrt{Q}x(0) ds\,.
\end{align}
Note that since 
\begin{equation}
    P=\begin{pmatrix}
        0 & 0\\0&I
    \end{pmatrix}\,,\hspace{0.1in} 
    Q=\begin{pmatrix}
       V & 0\\0&I
    \end{pmatrix}\,
    \hspace{0.1in}\text{and}
    \hspace{0.1in}
    A_R=\begin{pmatrix}
        0 & 0\\0&-R
    \end{pmatrix}\,,
\end{equation}
we have
\begin{equation}
    \sqrt{Q^{-1}}P\sqrt{Q^{-1}} = P\,,\hspace{0.1in}\text{and}\hspace{0.1in}\sqrt{Q}A_R\sqrt{Q^{-1}} = A_R\,.
\end{equation}
The above equation gives us
\begin{equation}
    \|\exp(\sqrt{Q}J\sqrt{Q}s+\sqrt{Q}A_R\sqrt{Q^{-1}}s)\|\leq \|\exp(A_Rs)\|\leq 1\,,
\end{equation}
where we have used the fact that $\sqrt{Q}J\sqrt{Q}$ is anti-Hermitian and that $A_R$ is negative definite. All of this gives us that the norm of the second term is
\begin{equation}
    \|x(0)\sqrt{Q}\|^2\|R\|t\,.
\end{equation}
Now, we have that 
\begin{equation}
    \|x(0)\sqrt{Q}\|^2 = x(0)^TQx(0) = q^TVq + \dot{q}^T\dot{q} = 2E\,,
\end{equation}
where $E$ is the total energy of the classical mechanical system. The norm of the third term \eq{third_term} and the fourth term \eq{fourth_term} can be bounded in a similar way and we get
\begin{equation}
   |K-K_R|\leq  2E\|R\|t + E\|R\|^2t^2\,.
\end{equation}
Therefore
\begin{equation}
    \frac{1}{E}|K-K_R|\leq 2\|R\|t(1+\frac{1}{2}\|R\|t)\,.
\end{equation}

Suppose there is a quantum algorithm that can compute an estimate $\hat{K}$ such that 
\begin{equation}
    |\hat{K} - \frac{K_R}{E}|\leq \epsilon\,,
\end{equation}
Then $\hat{K}$ also satisfies
\begin{equation}
    |\hat{K} - \frac{K}{E}|\leq |\hat{K} - \frac{K_R}{E}| + \frac{1}{E}|K-K_R| \leq \epsilon + 2\|R\|t(1+\frac{1}{2}\|R\|t) \,,
\end{equation}
Therefore, if 
\begin{equation}
    2\|R\|t(1+\frac{1}{2}\|R\|t)\leq \epsilon\,,
\end{equation}
i.e., if $\|R\|t\leq \epsilon/4$, then we have
\begin{equation}
    |\hat{K} - \frac{K}{E}|\leq 2\epsilon\,.
\end{equation}
Thus the estimate $\hat{K}$ also gives an additive approximation to the kinetic energy when $R=0$, which is a $\BQP$ hard problem by \thm{R_eq_0_hardness}.
\end{proof}

\section{Quantum algorithms for optimal control}\label{sec:optimal_control}
In this section, we present quantum algorithms for the vector and matrix Riccati equations. Like our quantum algorithm for the estimation of kinetic energy of damped mechanical systems, our quantum algorithm for the vector Riccati equation highlights the fact that it is sometimes useful to have the history state that gives access to the unnormalized solution state. Note that in this section, we reuse previous variables such as $x$, $y$, $A$ and $b$ with new definitions.

\subsection{Vector Riccati equation}
Recall the vector Riccati equation from \eq{Riccati_ivp} i.e.,
\begin{equation}
    \dot{y}=F_0+F_1y+yF_2y+yF_3\,,\hspace{1cm} y(0)=y_0\,,
\end{equation}
where $y$ and $F_0$ are $d$ dimensional column vectors, $F_1$ is a $d\times d$ matrix, $F_2$ is a $d$ dimensional row vector and $F_3$ is a complex number. Suppose the linear differential equation obtained by the M\"{o}bius transformation from \sec{Mobius} is
\begin{equation}\label{eq:LODE}
    \dot{x}=Ax+b\,,\hspace{1in}x(0)=x_0\,,
\end{equation}
where
\begin{equation}
    A=\begin{pmatrix}
        F_1&F_0\\F_2&F_3
    \end{pmatrix}\,,
\end{equation}
and 
\begin{equation}
    b=\begin{pmatrix}
        F_1w\\F_2w
    \end{pmatrix}\,,
\end{equation}
where $w$ is any time-independent vector. Here
\begin{equation}
    x=(u,v)^T\,,
\end{equation}
is such that $y=(u+w)v^{-1}$. We now show how to construct oracles to access the entries of $A$ and prepare a state proportional to $b$.
\begin{lemma}\label{lem:A_oracle}
Using the oracles for $F_i$, $i=0,1,2,3$, we can construct an oracle $O_A$ and $O_b$ with a constant overhead in gate complexity.
\end{lemma}
\begin{proof}
To get $O_A$, we can use the $F_i$ and apply them conditioned on the row and column index. Specifically, $O_A$ is 
\begin{equation}
O_A=\sum_{i,j=1}^{d-1}\ket{i,j}\bra{i,j}\otimes F_1 + \sum_{i=1}^{d-1}\ket{i,d}\bra{i,d}\otimes F_0 + \sum_{j=1}^{d-1}\ket{d,j}\bra{d,j}\otimes F_2 + \ket{d,d}\bra{d,d}\otimes F_3\,.
\end{equation}
To create $O_b$, we can take $\ket{w}$ to be the all-zeros basis state $\ket{0}$ and therefore
\begin{equation}
    b=\begin{pmatrix} e \\f\end{pmatrix}\,,
\end{equation}
where the $d$ dimensional vector $e$ and the number $f$ are
\begin{equation}
    e(i)=F_1(i,1)\,,\text{ and }\, f=F_2(1)\,.
\end{equation}
Since $F_i$ are assumed to be sparse with a constant number of nonzero entries in every row or column, $O_b$ can be constructed using standard techniques with constant gate complexity.
\end{proof}
To solve the above equation, we will use the result from \cite{krovi2022improved} stated below. Note that this version is to produce the history state and not the normalized solution state.
\begin{theorem}\label{thm:LODE}
Given an instance of a linear inhomogeneous equation \eq{LODE} where $A$ has sparsity $s$, dimension $d$, $C(A)=\sup_{t\in [0,T]}\exp(At)$ and oracles as given in \defn{oracles}, there exists a quantum algorithm that produces a quantum state proportional to  
\begin{equation}\label{eq:y}
    \ket{z}\sim \sum_{i=0}^{m-1}\ket{i,x_i} + \sum_{i=m}^{2m-1}\ket{i,x_m}\,,
\end{equation}
where $m=T\|A\|$ is the number of time steps, $x_i$ are $\epsilon$ approximate solutions at time steps $i=0,\dots, m$. Let
\begin{equation}
    k=O(\log T\|A\|(1+\frac{T\|b\|}{\|x(T)\|})\,.
\end{equation}
This algorithm succeeds with a constant success probability and with query complexity
\begin{equation}
    O\Big(mC(A)\poly\Big(\log d,s,k,\log(\frac{1}{\epsilon}),\log(mC(A))\Big)\Big)\,,
\end{equation}
and the gate complexity is greater by a factor of at most
\begin{equation}
    O\Big(\poly\Big(k,\log(\frac{1}{\epsilon}),\log(m)\Big)\Big)\,.
\end{equation}
\end{theorem}

The quantum algorithm to solve the Riccati equation \eq{Riccati_ivp} consists of the following steps. In the steps below, for any unnormalized vector $\phi$, the normalized version is denoted by $\ket{\phi}$.
\begin{enumerate}
    \item We take three registers, one for the time step, one for the Taylor series level and one for the solution of the differential equation (and other states such as the initial state). Define $z_{in}$ as
    \begin{equation}
        z_{in} = \ket{0,0,x_0} + h\sum_{i=1}^{m} \ket{i,1,b} + \sum_{i=m+1}^{2m}\ket{i,0,w}\,,
    \end{equation}
    where $\ket{w}=\ket{0}$, $x_0 = [y_0,1]^T$ (with $\ket{x_0}$ being the normalized version of $x_0$) and where $y_0$ is the initial state of the Riccati equation in \eq{Riccati_ivp} and $h=1/\|A\|$.
    Let us denote
    \begin{equation}
        \psi_{in}=\ket{0,0,x_0} + h\sum_{i=1}^{m} \ket{i,1,b}\,,
    \end{equation}
    and
    \begin{equation}
        w_{in} = \sum_{i=m+1}^{2m}\ket{i,0,w}\,.
    \end{equation}
    In the first step, we construct the state $\ket{z_{in}}$.
    \item Next, construct a linear system for the ODE \eq{LODE} from \cite{krovi2022improved} i.e.,
    \begin{equation}
        Lz=z_{in}\,,
    \end{equation}
    and solve the linear system to obtain the final state $\ket{z}$. This unnormalized state can be written as 
    \begin{equation}
        z=L^{-1}z_{in} = L^{-1}\psi_{in} + L^{-1}w_{in}\,.
    \end{equation}
    Note that $w_{in}=Lw_{in}$ since $L$ acts as the identity when the first register is in the support of the states $m+1$ to $2m$. Therefore
    \begin{equation}\label{eq:z}
        z=L^{-1}\psi_{in} + w_{in}=\sum_{i=1}^m\ket{i,0,x_i} + \sum_{i=m+1}^{2m}\ket{i,0,x_m} + \sum_{i=m+1}^{2m}\ket{i,0,w}\,.
    \end{equation}
    \item Measure the time register and post-select outcomes $m+1,\dots ,2m$ to get a state proportional to
    \begin{equation}
        \psi_T=x_m+w\,.
    \end{equation}
    \item Finally, perform the two outcome measurement of the computational basis with outcomes $P = \{1,\dots,d\}$ and $I-P = \{d+1\}$. Post-select on the outcome $k\in\{1,\dots, d\}$ (recall that $x_m$ is $d+1$ dimensional) to get the state $\ket{\hat{y}}$, which is a normalized version of the approximate solution $\hat{y}$. We show in the next section that $\hat y$ is the approximate version of the actual solution $y(t)$ of equation \eq{Riccati_ivp}. We can write this state as 
    \begin{equation}
        \ket{\hat{y}} = \frac{P\psi_T}{\|P\psi_T\|}\,,
    \end{equation}
    where $P$ denotes the orthogonal projector onto the first $d$ components of the $d+1$ dimensional state $\psi_T$.
\end{enumerate}    

Next we analyze all the aspects of the quantum algorithm for the Riccati equation such as its accuracy, success probability and the gate and query complexity of its circuit implementation.
\subsubsection{Solution error}
In this subsection, we bound the norm of the error between the normalized solution and the normalized approximate solution i.e,
\begin{equation}
    \|\ket{y}-\ket{\hat{y}}\|=\Bigg\|\frac{P(x+w)}{\|P(x+w)\|}-\frac{P(\psi_T)}{\|P(\psi_T)\|}\Bigg\|\,,
\end{equation}
where $P$ is the projector onto the components $1,\dots ,d$ of the $d+1$ dimensional vectors $x_m$, $w$ and $x$. We first bound the error between unnormalized versions of these states in the lemma below. First recall that $\|w\|=1$, $P(x+w)=u$ and $(I-P)(x+w)=v$, where $u$ and $v$ are from \eq{u_v}. We use this notation below. 
\begin{lemma}\label{lem:soln_error}
Suppose that the relative error in estimating $x$ is $\epsilon_1$ i.e.,
\begin{equation}
    \|x-\hat{x}\|\leq \epsilon_1\|x\|\,,
\end{equation}
then the error between $y$ and $\hat{y}$ can be bounded as
\begin{equation}
    \|y - \hat{y}\|\leq \epsilon\,,
\end{equation}
where
\begin{equation}
    \epsilon =2\epsilon_1(1+\frac{1}{\|y\|}+\frac{1}{\|u\|})\,,
\end{equation}
\end{lemma}
\begin{proof}
First consider
\begin{align}
    &\|P(x+w) - P(\hat x + w)\| \\
    &= \|Px -P\hat x \|\\
    &\leq \epsilon_1\|x\|\\
    &=\epsilon_1\|(x+w)-w\|\\
    &\leq \epsilon_1\Bigg(1+\frac{1}{\|y\|}+\frac{1}{\|u\|}\Bigg)\|u\|\,,
\end{align}
where in the last line, we used the fact that $\|w\|=1$ and $y=u v^{-1}$.

Using this, we can now bound $\|\ket{y}-\ket{\hat{y}}\|$ as follows.
\begin{align}
    \|\ket{y}-\ket{\hat{y}}\|&=\Bigg\|\frac{P(x+w)}{\|P(x+w)\|}-\frac{P(\psi_T)}{\|P(\psi_T)\|}\Bigg\|\\
    &=\Bigg\|\frac{P(x+w)}{\|P(x+w)\|}-\frac{P(\hat x + w)}{\|P(\hat x + w)\|}\Bigg\|\\
    &\leq \Bigg\|\frac{P(x+w)}{\|P(x+w)\|}-\frac{P(\hat x + w)}{\|P(x + w)\|}\Bigg\| + \Bigg\|\frac{P(\hat x + w)}{\|P(x+w)\|}-\frac{P(\hat x + w)}{\|P(\hat x + w)\|}\Bigg\|\\
    &\leq \frac{\epsilon}{2} + \|P(\hat x + w)\|\frac{\Big|\|P(\hat x + w)\| - \|P(x+w)\|\Big|}{\|P(\hat x + w)\|\|P(x+w)\|}\\
    &=\epsilon\,.
\end{align}
\end{proof}

\subsubsection{Probability of success}
\begin{lemma}\label{lem:success_prob}
The probability of obtaining the state $\hat{y}/\|y\|$, where $\hat{y}$ is the approximate solution to the Riccati equation is bounded below as follows.
\begin{equation}
    p_{success}\geq \frac{1}{108g^2}\,,
\end{equation}
where 
\begin{equation}
    g=\frac{\max_{t\in [0,T]}\|u(t)\|^2}{\|u(T)\|^2}\,.
\end{equation}
\end{lemma}
\begin{proof}
The probability of getting $\hat{y}$ can be broken down into two parts. The first part is the probability of getting the state $\psi_T$ and the second is the projector onto the first $d$ components. The overall probability can be written as follows.
\begin{equation}
    p_{\mathrm{success}}=\frac{m\|x_m+w\|^2}{\|z\|^2}\frac{\|P(x_m+w)\|^2}{\|x_m+w\|^2}=\frac{m\|P(x_m+w)\|^2}{\|z\|^2}\,,
\end{equation}
where we use $x_m$ for $\hat x$ from \eq{z}. We now upper bound the denominator. Recall from \eq{z} that
\begin{equation}
    \|z\|^2\leq \sum_{i=1}^m \|x_i\|^2 + m\|x_m+w\|^2\,.
\end{equation}
Using Theorem 8 from \cite{krovi2022improved}, we can bound the above quantity in terms of the maximum norm of the solution. Using this, we have
\begin{align}
    \|z\|^2&\leq m(1+\gamma)^2\max_{t\in [0,T]}\|x(t)\|^2 + m(1+\gamma)^2\|x(T)+w\|^2\\
    &\leq 4m(1+\gamma)^2\Big[\max_{t\in [0,T]}\|x(t)\|^2 +1\Big]\\
    &\leq 12m(1+\gamma)^2u_{max}^2\,.
\end{align}
The probability can now be bounded as (assuming $\epsilon,\gamma<1/2$)
\begin{align}
    p_{\mathrm{success}}&\geq \frac{m\|P(x_m+w)\|^2}{12m(1+\gamma)^2u_{max}^2}\\
    &\geq \frac{(1-\epsilon)^2\|u(T)\|^2}{12(1+\gamma)^2u_{max}^2}\\
    &\geq \frac{1}{108g^2}\,,
\end{align}
where
\begin{equation}
    g=\frac{\max_{t\in [0,T]}\|u(t)\|^2}{\|u(T)\|^2}\,.
\end{equation}

\end{proof}

\subsubsection{Main result}
Collecting these results together, we have our main result for Riccati equations. As pointed out in \sec{C(A)}, the result below gives a quantum algorithm for the Riccati nonlinear equation in the highly nonlinear regime.
\begin{theorem}
For the problem of producing a quantum state proportional to the solution of the Riccati differential equation described in \prb{Riccati_vector}, the quantum algorithm for the Riccati equation produces a quantum state proportional to the solution to within additive error $\epsilon$ with a constant success probability with gate complexity given by
\begin{equation}
    O\Big(g T\|A\|C(A)sk\polylog\Big(d,k,\frac{1}{\epsilon},T\|A\|C(A)\Big)\Big)\,,
\end{equation}
and the query complexity is
\begin{equation}
    O\Big(g T\|A\|C(A)sk\polylog\Big(k,\frac{1}{\epsilon},T\|A\|C(A)\Big)\Big)\,,
\end{equation}
where
\begin{align}
    &C(A) \leq \begin{cases}
			C_d(1+C_dT\|F_2\|), & \text{if $F_0=0$ and $F_2\neq 0$}\\
            C_d(1+C_dT\|F_0\|), & \text{if $F_2=0$ and $F_0\neq 0$}\\
            \exp(\mu + \|F_0\| + \|F_2\|), & \text{if both nonzero}
		 \end{cases}\\
    &\|A\| \leq \max(\|F_0\|+\|F_1\|,\|F_2\|+\|F_3\|)\,,
\end{align}
where $\mu$ is the log-norm of $F_1$ and $C_d$ is defined in \eq{C_d} and
\begin{equation}
    k=O(\log (1+\frac{\|b\|Te^2}{\|x(T)\|}))\,.
\end{equation}
\end{theorem}
\begin{proof}
We discuss the construction of the initial state first. In a manner similar to \cite{BCOW17, krovi2022improved}, we first construct the state
\begin{equation}
    \frac{\|x_0\|}{N_{\mathrm{init}}}\ket{0,0,0,0} + \frac{\sqrt{m}h}{N_{\mathrm{init}}}\ket{0,1,0,0} + \frac{\sqrt{m}}{N_{\mathrm{init}}}\ket{0,0,1,0}\,,
\end{equation}
where 
\begin{equation}
    N_{\mathrm{init}}=\sqrt{\|x_0\|^2 + m h^2\|b\|^2 + m}\,.
\end{equation}
Next, controlled on the second and third registers being in $00$ and $10$ respectively, we apply $O_x$ and $O_b$ to get
\begin{equation}
    \frac{\|x_0\|}{N_{\mathrm{init}}}\ket{0,0,0,\bar{x}_0} + \frac{\sqrt{m}h}{N_{\mathrm{init}}}\ket{0,1,0,\bar{b}} + \frac{\sqrt{m}}{N_{\mathrm{init}}}\ket{0,0,1,0}\,.
\end{equation}
Now conditioned on the second and third registers being in $10$ apply the rotation on the first register that takes
\begin{equation}
    \ket{0}\rightarrow \frac{1}{\sqrt{m}}\sum_{i=0}^{m-1}\ket{i}\,,
\end{equation}
and conditioned on the second and third registers being in $01$, apply the rotation
\begin{equation}
    \ket{0}\rightarrow \frac{1}{\sqrt{m}}\sum_{i=m}^{2m-1}\ket{i}\,.
\end{equation}
This produces the initial state
\begin{equation}
    \ket{z_0}=\frac{1}{N_{\mathrm{init}}}\Big[\ket{0,0,x_0} + h\sum_{i=1}^{m-1} \ket{i,1,b} + \sum_{i=m+1}^{2m}\ket{i,0,w}\Big]\,.
\end{equation}
The next step of the algorithm is to implement the quantum algorithm described in \cite{krovi2022improved}. The implementation complexity of that algorithm is proved in \cite{krovi2022improved}. The bound on $\|A\|$ comes from its block structure and the bound on $C(A)$ comes from \lem{C(A)_F_0} and \lem{C(A)_F_neq_0}. This gives the query and gate complexities in the theorem.
\end{proof}

\subsection{Matrix Riccati equation}\label{sec:mre_imp}
In this subsection, we give an algorithm to implement the solution of the matrix Riccati equation described in \eq{Riccati_ivp}. We assume that we have a block encoded version of the linearized matrix $A$ from \eq{linear_Riccati}. This can be obtained from oracle access to the matrices $F_i$ using \lem{A_oracle}. Our goal is to create a block encoding of the solution $y(t)$, which is an $N\times M$ matrix. The algorithm consists of the following steps.
\begin{enumerate}
    \item Following \cite{krovi2022improved}, we first create a block encoded unitary that encodes the linear system that produces the solution of the linear ODE from \eq{linear_Riccati} and invert the linear system to give a unitary $U_1$ (which is a $(\alpha_1,a_1,\epsilon_1)$ block encoding of $L^{-1}$). More precisely, we create the block encoded matrix $L$ from \cite{krovi2022improved}, which is given by
    \begin{equation}
        L=I-N\,,
    \end{equation}
    where
    \begin{align}
        &N=M_1(I-M_2)^{-1}\\
        &M_1=\sum_{j=0}^{k-1}\ket{j+1}\bra{j}\otimes \frac{Ah}{j+1}\\
        &M_2=\sum_{j=0}^k\ket{0}\bra{j}\otimes I\,.
    \end{align}
    It was shown in \cite{krovi2022improved} that $L^{-1}$ acts on the initial state $z'_{in}$ and produces a state $z'$ that encodes the solution of the linear ODE i.e.,
    \begin{equation}
        L^{-1}z'_{in}=z'\,,
    \end{equation}
    where 
    \begin{equation}
        z'_{in}=\ket{0,0,x_0} + h\sum_{i=0}^{m-1}\ket{i,1,b'}\,.
    \end{equation}
    In the equation above and in \cite{krovi2022improved}, $b'$ is a vector which can be prepared efficiently as a normalized quantum state. 
    
    \item In our case, the initial state $x_0$ and the constant term in the linear ODE $b$ are $(N+M)\times M$ matrices. Recall that $b$ is the matrix
    \begin{equation}
        b=\begin{pmatrix}e\\f\end{pmatrix}=\begin{pmatrix}F_1w\\F_2w\end{pmatrix}\,,
    \end{equation}
    where $w$ can be chosen to be any time-independent matrix of the same dimension as $x_0$.

    To be able to use the algorithm from \cite{krovi2022improved}, we encode these as block unitary matrices as follows i.e., we want to block-encode the matrix
    \begin{equation}
        z_{in}=\ket{0,0}\bra{0,0}\otimes x_0 + h\sum_{i=0}^{m-1}\ket{i,1}\bra{0,0}\otimes b\,.
    \end{equation}
    Every column of $z_{in}$ when acted upon by $L^{-1}$ gives a column of the solution matrix. To implement $z_{in}$ as a block unitary, we can use sparse matrix encoding to first block-encode
    \begin{equation}
        \ket{0,0}\bra{0,0}\otimes x_0 + \sqrt{m}h\ket{0,1}\bra{0,0}\otimes b\,.
    \end{equation}
    Next we implement a unitary on the left taking
    \begin{equation}
        \ket{0,0}\rightarrow\ket{0,0}\hspace{0.1in}\text{and}\hspace{0.1in}\ket{0,1}\rightarrow\frac{1}{\sqrt{m}}\sum_{i=0}^{m-1}\ket{i,1}\,.
    \end{equation}
    Denote this block encoded unitary $U_2$ with parameters $(\alpha_2,a_2,\epsilon_2)$. 
    
    \item Next we use the block-encoding of $L^{-1}$ and a block encoding of $z_{in}$ to get a block encoding of the solution i.e., $L^{-1}z_{in}$. Denote by $U_3$ the block encoding of $L^{-1}z_{in}$ obtained by multiplying the block encodings of $L^{-1}$ and $z_{in}$. Let $U_3$ be a $(\alpha_3,a_3,\epsilon_3)$ block encoding of $L^{-1}z_{in}$.
    
    \item We now have a unitary that block encodes the solution $x_T$. This can be written as 
    \begin{equation}
        x_T=\begin{pmatrix}u(T)\\v(T) \end{pmatrix}\,.
    \end{equation}
    Using matrix arithmetic, we can block encode $x_T+w$ as $U_4$ with parameters $(\alpha_4,a_4,\epsilon_4)$.
    \item Now conditioned on the lower block, we can invert $v$ and create a block unitary $U_5$. Assume that it is a $(\alpha_5,a_5,\epsilon_5)$ block encoding of $v^{-1}$.
    \item Finally, we multiply $U_4$ and $U_5$ to create a block version $U_6=U_4U_5$ of the solution $y_T=(u+w) v^{-1}$ with parameters $(\alpha_6,a_6,\epsilon_6)$.
\end{enumerate}

We compute the block-encoding triples and analyze the algorithm below. 
\subsubsection{Analysis of the algorithm}
In this subsection, we show that the algorithm above produces a solution for the matrix Riccati equation. Specifically, we prove the following.
\begin{theorem}
For the matrix Riccati equation defined in \eq{Riccati_ivp}, let $\kappa_V$ be the condition number of $V(t)$ coming from the Jacobi condition. Then, the above algorithm produces a matrix encoded as a block unitary, which is $\epsilon$ close to the solution with query and gate complexity given by
\begin{equation}
    O(s\kappa_V\kappa_L\polylog(\frac{1}{\epsilon},s,\kappa_V,\kappa_L,M+N))\,.
\end{equation}
where
\begin{align}
    &C(A) \leq \begin{cases}
			C_d(1+C_dT\|F_2\|), & \text{if $F_0=0$ and $F_2\neq 0$}\\
            C_d(1+C_dT\|F_0\|), & \text{if $F_2=0$ and $F_0\neq 0$}\\
            \exp(\mu + \|F_0\| + \|F_2\|), & \text{if both nonzero}
		 \end{cases}\\
    &\|A\| \leq \max(\|F_0\|+\|F_1\|,\|F_2\|+\|F_3\|)\,,
\end{align}
where the log-norm of $F_1$ is denoted $\mu$ and
\begin{equation}
    C_d=\max\{\sup_{t\in [0,T]}\|\exp(F_1t)\|,\sup_{t\in [0,T]}\|\exp(F_3t)\|\}\,,
\end{equation}
and 
\begin{equation}
    \kappa_L=O(C(A)T\|A\|)\,.
\end{equation}
\end{theorem}
\begin{proof}
We now bound the query and gate complexities of the algorithm to produce a state $\epsilon$ close to the solution. From \cite{krovi2022improved}, we have that the condition number of $L$ is $\kappa_L=O(C(A)m)$, where $m=T\|A\|$.
\begin{enumerate}
    \item To implement $U_1$, which is a block encoding of $L^{-1}$, the query and gate complexities are given by \lem{matrix_inversion}. If we block encode $L$ as a $(s,\log(m(N+M)),\delta)$ block encoding, then $L^{-1}$ is a $(\alpha_1,a_1,\epsilon_1)$ block encoding, where
    \begin{align}
        &\alpha_1 = 2\kappa_L\\
        &a_1=\log (m(N+M)) + O(\log (\kappa_L^2\log 1/\epsilon))\,.
    \end{align}
    We also need to pick $\delta$ according to \lem{matrix_inversion} as
    \begin{equation}
        \delta = o(\epsilon/\kappa_L^2\log^3(\kappa_L^2/\epsilon))\,.
    \end{equation}
    The gate complexity is
    \begin{equation}
        T_1=O(s\kappa_L(\log (m(N+M))+T_U)\log^2(\kappa_L^2/\epsilon))\,,
    \end{equation}
    where the cost of block encoding $L$ is
    \begin{equation}
        T_U=O(\log(m(N+M)) + \log^{2.5}(s^2/\delta))\,.
    \end{equation}
    
    \item To implement $U_2$ as a $(\alpha_2,a_2,\epsilon_2)$ block encoding of $z_{in}$, we can use \lem{sparse_matrix}. The parameters are
    \begin{align}
        &\alpha_2=s\\
        &a_2 = \log (m(N+M))\,.
    \end{align}
    The query complexity is $O(1)$ and the gate complexity is $O(\log(m(N+M))+\log^{2.5}(s^2/\epsilon_2))$.
    \item Multiplying two block encoded unitaries can be done using \lem{matrix_arithmetics}. The complexity of implementing $U_3$ is $O(T_1+T_2)$ and the parameters are $(\alpha_3,a_3,\epsilon_3)$, where
    \begin{align}
        &\alpha_3=\alpha_1\alpha_2\\
        &a_3=a_1+a_2\\
        &\epsilon_3=\epsilon_1\alpha_2+\epsilon_2\alpha_1\,.
    \end{align}
    \item This step can be done using matrix arithmetic. A block encoding of $w$ with parameters $(1,\log m(N+M),\epsilon)$ can be done in $T_w=O(\log(m(N+M))+\log^{2.5}(1/\epsilon))$ gates. Therefore, the gate complexity of block encoding $u+w$ is $T_4=O(T_3+T_w)$.
    \item Inverting $v$ can be done using \lem{matrix_inversion}. We can get a $(\alpha_5,a_5,\epsilon_5)$ block encoding of $v^{-1}$ where
    \begin{align}
        &\alpha_5=2\kappa_V\\
        &a_5 = O(\log M + \log(\kappa_V^2\log(1/\epsilon_5)))\,.
    \end{align}
    The gate complexity is 
    \begin{equation}
        T_5 = O(\alpha_4\kappa_V(\log M + T_4)\log^2(\kappa_V^2/\epsilon))\,.
    \end{equation}
    \item Finally, we multiply the block encoded unitaries of $u+w$ and $v^{-1}$ to get a $(\alpha_6,a_6,\epsilon_6)$ block encoding of $y(T)$, where
    \begin{align}
        &\alpha_6 = \alpha_5\alpha_4\\
        &a_6 = a_5+a_4\\
        &\epsilon_6=\alpha_4\epsilon_5 + \alpha_5\epsilon_4\,.
    \end{align}
    The gate complexity is $T_6=O(T_5+T_4)$.
\end{enumerate}
Putting it together, we get a 
\begin{equation}
    (2\kappa_V(2\kappa_Ls+1),\polylog(\kappa_L,\kappa_V,1/\epsilon)+\log(M+N),\epsilon)
\end{equation}
block-encoding of the solution with query complexity
\begin{equation}
    O(s\kappa_V\kappa_L\polylog(\frac{1}{\epsilon},s,\kappa_V,\kappa_L,M+N))\,.
\end{equation}
\end{proof}

\subsection{Applications}
\subsubsection{Linear quadratic regulator}\label{sec:matrix_Riccati}
In this section, we describe how to implement the solution to the matrix Riccati equation as a block encoded unitary. Recall from \prop{Riccati_ODE} that the matrix Ricatti equation can be linearized in a larger dimension. More specifically, we can define an $(N+M)\times (N+M)$ matrix
\begin{equation}
    A=\begin{pmatrix}F_1 & F_0\\F_2 & F_3 \end{pmatrix}\,,
\end{equation}
such that if the solution of the linear ODE
\begin{equation}
    \dot{x}=Ax+b\,,
\end{equation}
is
\begin{equation}
    x(t) = \begin{pmatrix}u(t)\\v(t) \end{pmatrix}\,,
\end{equation}
then the solution of the matrix Riccati equation
is given by 
\begin{equation}
    y(t)=u(t)v^{-1}(t)\,.
\end{equation}
The initial condition for the linear ODE is
\begin{equation}
    x(0) = \begin{pmatrix}y_0\\I \end{pmatrix}\,,
\end{equation}
where $y_0$ is the initial condition of the Riccati equation. Note that $x$ is an $(N+M)\times M$ matrix and $y$ is an $N\times M$ matrix.

We now explain how to solve a boundary value problem (which is needed for the LQR problem) using the above techniques. First, we need the following lemma.
\begin{lemma}\label{lem:lin-ode-reverse}
For the linear ODE $\dot{x}=Ax+b$, if at time $t_f$, the solution is $x_f$ obtained from an initial condition $x_0$, then we can to $x_0$ from the initial condition $x_f$ by time evolving the equation $\dot{x}=-Ax-b$.
\end{lemma}
\begin{proof}
We have that 
\begin{equation}
    x_f=\exp(At_f) x_0 + \int_0^{t_f}\exp(At)b dt\,. 
\end{equation}
Using this it can be seen that
\begin{equation}
    x_0=\exp(-At_f)x_f -\exp(-At_f)\int_0^{t_f}\exp(At)b dt=\exp(-At_f)x_f +\int_0^{t_f}\exp(-At)(-b) dt\,,
\end{equation}
where we have used the invariance of the integral under the change of variable $t\rightarrow t_f-t$.
\end{proof}

Suppose we need to solve a boundary value problem given by the condition $y_{t_f}=y_f$, then we can first create a matrix $x_f$ as follows.
\begin{equation}
    x_f=\begin{pmatrix}y_f\\I
    \end{pmatrix}\,.
\end{equation}
We can run the algorithm for a linear ODE by block encoding $-A$ and $-b$. This will give us a block-encoding of $x_0$ by conditioning on the time register. More specifically, we get the following block-encoded unitary
\begin{equation}
    L^{-1}z_{in} = \sum_{i=1}^m \ket{i,0}\bra{0,0}\otimes y_i\,,
\end{equation}
where $y_i$ is the block-encoding of the solution at different times. At time $m=T\|A\|$, we have the block-encoding of the approximate solution $y_m$. By applying a shift operator on the first register, we can bring it into a standard block-encoding as in \defn{block_encoding}. We can then use this block-encoding of $y(0)$ to run the algorithm forwards i.e., now by block-encoding $A$ and $b$ to the final time $T$ to get a block encoding of $y(T)$. 

In the next subsection, we describe the algorithm to solve the Hamilton-Jacobi-Bellman (HJB) equation using its connection to the matrix Riccati equation, which produces a block encoding of a solution to the matrix Riccati equation.

\subsubsection{Quantum algorithm for the linear quadratic regulator problem}\label{sec:hjb_alg}
In this subsection, we will use the results on matrix Riccati equation from above to give a quantum algorithm to solve the linear quadratic regulator problem. For this problem, let us denote
\begin{equation}
    F_0=Q\,,F_1=F\,,F_2=GR^{-1}G^\dag\,,F_3=F^\dag\,,
\end{equation}
with their dimension being $N$ and sparsity $s$.
\begin{theorem}\label{thm:HJB_algorithm}
There exists a quantum algorithm that produces a quantum solution to the linear quadratic regulator, which is a special case of the HJB equation \eq{HJB_special}, with query and gate complexity
\begin{equation}
    O(s\kappa_V\kappa_L\polylog(\frac{1}{\epsilon},s,\kappa_V,\kappa_L,N))
\end{equation}
where (as before)
\begin{align}
    &C(A) \leq \begin{cases}
			C_d(1+C_dT\|F_2\|), & \text{if $F_0=0$ and $F_2\neq 0$}\\
            C_d(1+C_dT\|F_0\|), & \text{if $F_2=0$ and $F_0\neq 0$}\\
            \exp(\mu + \|F_0\| + \|F_2\|), & \text{if both nonzero}
		 \end{cases}\\
    &\|A\| \leq \max(\|F_0\|+\|F_1\|,\|F_2\|+\|F_3\|)\,,
\end{align}
where the log-norm of $F_1$ is denoted $\mu$ and
\begin{equation}
    C_d=\max\{\sup_{t\in [0,T]}\|\exp(F_1t)\|,\sup_{t\in [0,T]}\|\exp(F_3t)\|\}\,,
\end{equation}
and 
\begin{equation}
    \kappa_L=O(C(A)T\|A\|)\,.
\end{equation}
\end{theorem}
\begin{proof}
The quantum algorithm essentially constructs a block-encoding of the solution $P$ of the matrix Riccati equation defined in \eq{mre-hjb} and then uses \eq{mjb-soln} to construct the solution. We discuss these steps next.

First, using oracle access to the matrices $R,Q,F$ and $G$, we can construct block-encoded unitaries corresponding to these matrices. As above, let us denote
\begin{equation}
    F_0=Q\,,F_1=F\,,F_2=GR^{-1}G^\dag\,,F_3=F^\dag\,,
\end{equation}
and denote their block-encoded unitary versions $U_0$ through $U_3$. To construct the matrix $A$ defined in \eq{linear_Riccati}, we conditionally apply each of the unitaries $U_i$ i.e., first create
\begin{equation}
    V_1=\ket{0}\bra{0}\otimes U_1 + \ket{1}\bra{1}\otimes U_3\,,
\end{equation}
then create the off-diagonals as follows.
\begin{equation}
    V_2=(X\otimes I)(\ket{0}\bra{0}\otimes U_2 + \ket{1}\bra{1}\otimes U_0)=\ket{1}\bra{0}\otimes U_2 + \ket{0}\bra{1}\otimes U_0\,.
\end{equation}
Adding these two block matrices i.e., $U=V_1+V_2$, we get a block encoded version of the linear operator $A$. This takes $O(s)$ queries, where $s$ is the maximum sparsity of the matrices $R,Q, F$ and $G$.

Next, we use block-encoded implementation of matrix Riccati algorithm from \sec{mre_imp} to create a block-encoded version of $P$ from \eq{mre-hjb}.

Then, we create a block encoding of $R^{-1}$ and $G^\dag$. Using these and the oracle for $x(t)$, we can create the state
\begin{equation}
    \bar{u}(t)=-R^{-1}G^\dag Px(t)\,.
\end{equation}
The query and gate complexities follow from the complexity of the matrix Riccati algorithm and the complexity of matrix arithmetic.
\end{proof}

\section{Conclusions and open questions}\label{sec:conclusions}
We have presented quantum algorithms to simulate realistic mechanical systems such as ones that satisfy the following differential equation.
\begin{equation}
    M\ddot{q}+R\dot{q}+Vq=s\,,
\end{equation}
where $q(t)$ is the position vector of the mechanical system. This includes damped coupled oscillators extending the work of \cite{babbush2023exponential}. We have shown that when the damping strength is inverse polynomial in the number of qubits, then the problem of estimating the kinetic energy of the system is $\BQP$ hard. This means that unless $\BPP=\BQP$, no classical algorithm can solve this problem efficiently. 

We then give quantum algorithms to simulate the vector and matrix Riccati equation. Using the algorithm for a matrix Riccati equation, one can solve the linear quadratic regulator problem. The LQR problem is a fundamental problem in control theory and involves a quadratic nonlinearity. Our quantum algorithm can efficiently solve a quadratic nonlinear differential equation for a large nonlinearity ($R\gg 1$ in the language of \cite{Liue2026805118}). To our knowledge, this is the first work that gives an efficient algorithm to solve a differential equation with $R>1$ with rigorous performance bounds. Aside from the LQR problem, the Riccati equation is used in predicting the stability of magnetohydrodynamics (MHD) \cite{Glasser}.

We list a few open problems that arise from this work. First, applying the algorithm to estimate kinetic energy of damped coupled oscillators to estimate other quantities of interest would be useful. Second, the quantum algorithm for damped coupled oscillators is efficient for damping strengths that are $\polylog(d)$, where $d$ is the dimension. However, the hardness result requires that the strength be less than some inverse polynomial. Can it be improved to a constant damping strength? A related question is if there are classical algorithms that can estimate the kinetic energy in the presence of large enough damping strength. Third, are there other end-to-end applications of the quantum algorithm for the matrix Riccati equation? Another important direction would be to generate accurate resource for end-to-end applications like the stability of magnetohydrodynamics for instances at utility scale.

\section{Acknowledgments}
This material is based upon work supported by the U.S. Department of Energy, Office of Science, Office of Fusion Energy Sciences, under Award Number DE-SC0020264 and work supported by the Defense Advanced Research Projects Agency (DARPA) under Contract No. HR001122C0063.



\appendix
\section*{Appendices}
\section{Auxiliary lemmas}\label{appx:auxil_lemma}
\begin{lemma}\label{lem:q_i_diff}
When $b=0$ in \eq{M_R_V_b_eq}, the quantum algorithm from \sec{class_ham_sim} outputs $\hat{\dot{q}}$ such that we have
\begin{equation}
    |\dot{q}_i - \hat{\dot{q}}_i|\leq \epsilon\dot{q}_i\,,
\end{equation}
for some constant $\epsilon$.
\end{lemma}
\begin{proof}
Let $x=(q,\dot{q})^T$ and let $\Pi_i$ be the projector onto the $i^{th}$ component of $x$. From the proof of Theorem 3 of \cite{krovi2022improved}, we have
\begin{equation}
    \|\Pi_i(x(T)-\hat{x})\|\leq \frac{me^3}{(k+1)!}\|\Pi_ix(T)\|\leq \epsilon x_i\,,
\end{equation}
where we assume that we have chosen $k$ as in \cite{krovi2022improved} such that 
\begin{equation}
  \frac{me^3}{(k+1)!}\leq \epsilon\,.  
\end{equation}
By choosing $i$ to correspond to $\dot{q}_i$, we get the statement in the lemma.
\end{proof}

\bibliographystyle{ieeetr}
\bibliography{qdiff}

\end{document}